\documentclass[reprint, aps, prl, superscriptaddress]{revtex4-2}
\usepackage{graphicx, float, xcolor}
\graphicspath{{fig/}{../../figs/fi_fig/}{../../figs/recover/}}
\usepackage[ruled,linesnumbered]{algorithm2e}
\usepackage{physics}
\usepackage{hyperref}
\usepackage{amsthm, amssymb, bm}
\usepackage[T1]{fontenc}
\usepackage{xcolor}

\definecolor{revisioncolor}{RGB}{0,90,180}
\newcommand{\rev}[1]{\textcolor{black}{#1}}

\let\originalsection\section

\renewcommand{\section}[1]{\par\medskip\noindent\textbf{{\textit{#1.}---}}}
\newcommand{\ii}{\mathrm{i}}
\newcommand{\myqc}{\text{HP}}
\newcommand{\uqc}[1]{U_{\mathrm{HP}}^{(#1)}}
\newcommand{\cpgate}[3]{\mathrm{CP}_{#1, #2}(#3)}
\newcommand{\fii}{\mathrm{DFI}}
\newcommand{\Lambdacomp}[1]{\bar{\Lambda}_{\leq #1}}

\newtheorem{definition}{Definition}
\newtheorem{theorem}{Theorem}
\newtheorem{proposition}{Proposition}
\newtheorem{conjecture}{Conjecture}
\newtheorem*{restatedtheorem}{Theorem}
\newtheorem*{restatedproposition}{Proposition}

\begin{document}
\title{$\mathcal{O}(n)$ alternative to Quantum Fourier Transform with efficient neural net classical post-processing}

\author{Kaiming Bian}
\thanks{These two authors contributed equally.}
\affiliation{Shenzhen Institute for Quantum Science and Engineering, Southern University of Science and Technology, Shenzhen, 518055, China}
\affiliation{International Quantum Academy, Shenzhen, 518048, China}

\author{Zujin Wen}
\thanks{These two authors contributed equally.}
\affiliation{Department of Physics, City University of Hong Kong, Tat Chee Avenue, Kowloon, Hong Kong SAR, China}
\author{Oscar Dahlsten}
\affiliation{Department of Physics, City University of Hong Kong, Tat Chee Avenue, Kowloon, Hong Kong SAR, China}
\date{\today}

\begin{abstract}
The Quantum Fourier Transform (QFT) is \rev{employed} by hidden subgroup problem (HSP) algorithms, including Shor's algorithm for factoring. The circuit depth of the QFT remains challenging for near-term hardware. To find shallower alternatives we identify two properties that are exploited by the QFT to enable HSP. Firstly, the shift invariance of the QFT allows for the removal of a random overall shift. Secondly, the QFT retains information about the hidden subgroup generator accessible in the measurement outcomes. We quantify that information via the discrete Fisher information. We construct a family of shallow circuits using Hadamards and controlled-Phase gates, HP-$L$ circuits, that we prove preserve shift invariance. Numerical analysis shows these circuits retain exponentially growing Fisher information. \rev{The $\mathcal{O}(n)$ HP-$1$ is employed in place of the $\mathcal{O}(n^2)$ QFT in our numerical implementation of Shor's algorithm.} An efficient neural network is used for the corresponding classical post-processing.
\end{abstract}

\maketitle

\section{Introduction}The $\mathcal{O}(n^2)$ circuit depth required by the standard Quantum Fourier Transform (QFT) presents experimental challenges. The QFT is used in the exponential quantum speedups of landmark protocols—such as Shor's algorithms for prime factorization and discrete logarithms \cite{Simon1997, Shor1994, Shor1997}. The associated family of algorithms, all using the QFT, is known as the (quantum) hidden subgroup problem (HSP) algorithms \cite{Kitaev1995, Jozsa2001, Lomont2004, liu2024information}. The experimental implementation of the QFT remains challenging. For a $100$-qubit system ($n=100$), the QFT demands thousands of circuit layers, which vastly exceeds current hardware capabilities. For instance, IBM's $127$-qubit Eagle processor is limited to a depth of $60$ layers even with extensive error mitigation \cite{kim2023evidence}; Google's $105$-qubit Willow processor reaches its operational threshold at around $40$ layers~\cite{google2025quantum}; USTC's $105$-qubit Zuchongzhi~3.0 processor demonstrated random circuit sampling on $83$-qubits, with around $32$ layers~\cite{Zuchongzhi}. Moreover, for quantum error-mitigation protocols that are commonly used, the sampling overhead can scale exponentially with the circuit depth~\cite{takagi2022fundamental, takagi2023universal}.

There has been significant progress in reducing the circuit depth required. These efforts have enabled record experimental QFT fidelities \cite{aumann2026demonstrating} through techniques including \rev{(i) approximate or truncated QFT constructions, including those that omit small-angle controlled rotations \cite{Coppersmith2002, park2023reducing}, as well as $\mathcal{O}(n\log n)$-step algorithms for approximating the QFT over cyclic groups $\mathbb{Z}_p$, where $n=\log p$ for arbitrary $p$\cite{HalesHallgren2000}; (ii) parallel circuit constructions, which implement an approximate QFT over $\mathbb{Z}_{2^n}$ with error $\epsilon$ in depth $\mathcal{O}(\log n+\log\log(1/\epsilon))$ \cite{CleveWatrous2000}; }(iii) dynamic mid-circuit measurements, which reduce the resource scaling of QFT followed immediately by measurement from \(\mathcal{O}(n^2)\) two-qubit gates to \(\mathcal{O}(n)\) mid-circuit measurements without connectivity constraints \cite{baumer2024quantum}, and (iv) quantum circuit cutting and cluster-simulation methods, which trade coherent quantum depth for additional sampling and classical post-processing overhead~\cite{tang2021cutqc, Peng2020}. 

We here take a different approach to reducing the circuit depth required. Rather than approximate the QFT we identify key features of the QFT that make it compatible with the HSP. We focus on shift-invariance combined with retention of information about the generator in question. Fig.~\ref{fig:illus}(a) summarizes the main idea. We design non-QFT circuits that retain those key features.

\begin{figure}[htbp]
    \centering
    \includegraphics[width=0.8\columnwidth]{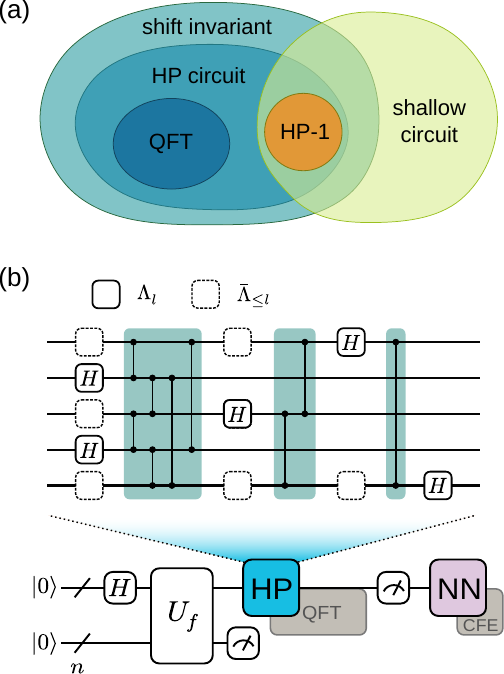}
    \caption{\label{fig:illus} \textbf{Shallow shift-invariant alternative to the QFT.} (a) Shift invariance is a critical property of the QFT. We design a circuit family, denoted as \myqc-$L$, that preserves the shift invariance. Within this broader family, \myqc-1 specifically serves as a shallow-depth circuit that is experiment-friendly. (b) In the proposed HSP algorithmic routine, the deep QFT block (the left gray gate) is directly replaced by the \myqc~circuit. A neural network then processes the measurement outcomes, bypassing the standard continued fraction expansion typically used in the standard HSP. The HP circuit is composed of Hadamard and phase gates, with $\Lambda_l$ referring to the qubits with Hadamards applied in layer $l$. }
\end{figure}

We construct a circuit family, called \myqc-$L$ circuits, that preserves shift invariance. In particular, the \myqc-$1$ circuits are shallow and retain information about the generator. The \myqc-$L$ circuits interleave transversal Hadamard (H) gates with $L$ layers of phase (P) blocks. We prove analytically that \myqc-$L$ circuits are shift-invariant \rev{ under $\mathbb{Z}_{2^n}$}. Numerical scaling analysis demonstrates exponentially increasing retention of Fisher information $\sim e^{kn}$, albeit with a worse exponent $k$ than QFT. That $\sim e^{kn}$-scaling consistent with efficient scaling of the number of samples required. 

We numerically simulate replacing the QFT in Shor's algorithm with \myqc-1 circuits, and using an efficiently scaling neural network~\cite{zaheer2017deep} for classical post-processing. Factorisation is successfully achieved.  

\section{Necessary condition for \rev{HSP:} Shift Invariance}\rev{The standard quantum HSP algorithm~\cite{Shor1994, Shor1997, Simon1997, Jozsa2001} is illustrated in Fig.~\ref{fig:illus}(b) with gray blocks. The elements $x$ belong to a finite abelian group $G$ with group operation denoted as +. We are guaranteed that $f(x+r)=f(x)$ for the unknown  $r$ for all $x$, and thus that $f(x+r+r+...)=f(x)$. (In other words, the function $f(\cdot)$ is constant on cosets of the subgroup generated by $r$.) $f(\cdot)$ is said to {\em hide} $r$. The task is to find $r$.  After $U_f$ (Fig.~\ref{fig:illus}(b)), the state is proportional to $\sum_x \ket{x}\ket{f(x)}$.  Upon measuring (or imagining measuring) the second register and randomly finding $\ket{f(c)}$ for some $c$ the state becomes proportional to $\sum_{q\in \{0,1,2...\}} \ket{c+qr}\ket{f(c)}$. The shift $c$ then acts as a random (coset) offset. Removing this $c$ from the computational basis statistics is a key step in the quantum HSP, and is normally effected by applying the (inverse) quantum fourier transform (QFT). The output of the QFT, in terms of the statistics in the computational basis, is invariant under shifts by c, as will be detailed below. Additional details about the HSP algorithm are provided in Appendix~\ref{append: hsp and fs}.}

We now define shift invariance~\cite{NielsenChuang2000} for HSP sampling and derive a necessary condition that any shift-invariant unitary must satisfy. 

\begin{definition}[Shift invariance]
    A unitary $U$ is shift-invariant for HSP sampling if
    \begin{equation}
    \label{eq:shift invariant def}
        \left|\bra{x}U\ket{V}\right|^2
        =
        \left|\bra{x}U \ket{c+V}\right|^2,
        ~~
        \forall x,c \in G,\ \forall V \le G,
    \end{equation}
\rev{where, up to normalisation, $\ket{V} :=\sum_{q\in \{0,1,2...\}} \ket{qr}$  and $\ket{V+c} := \sum_{q\in \{0,1,2...\}} \ket{qr+c}$.}

\end{definition}

To identify shift-invariant circuits beyond the QFT, we first establish a necessary condition for shift invariance.
\begin{proposition}[Necessary condition of shift invariance]
\label{thm:necessary_modulus_condition}
    Suppose $U$ is a unitary on the group register that satisfies the shift-invariance condition above. Then, for every $x,y\in G$ there exist phases $\theta_{xy}\in\mathbb{R}$ such that
    \begin{equation}
        \bra{x}U\ket{y}=\frac{1}{\sqrt{|G|}}e^{\ii\theta_{xy}}.
    \label{eq:necessarycondition}
    \end{equation}
    
\end{proposition}

\begin{proof}[Proof sketch of Proposition~\ref{thm:necessary_modulus_condition}]
The proof hinges on analyzing the shift-invariance condition for the trivial subgroup $V = \{0\}$. For $V = \{0\}$, the state $\ket{V}$ is the computational basis state $\ket{0}$, and the shifted state is $\ket{c+V} = \ket{c}$. The requirement (Eq.\eqref{eq:shift invariant def}) that $\left|\bra{x}U\ket{0}\right|^2 = \left|\bra{x}U\ket{c}\right|^2$ for all $x, c \in G$ implies that, for any fixed output $x$, the magnitude of every matrix element in the $x$-th row of $U$ must be identical. The special case induces the uniform magnitude $\bra{x}U\ket{y} = 1/\sqrt{|G|} e^{\ii \theta_{xy}}$. The full details of this argument are provided in Appendix~\ref{app:necessary condition of sinv}.
\end{proof}

\section{L-phase-block Hadamard-Phase circuits (\myqc): A shift-invariant circuit family}Proposition \ref{thm:necessary_modulus_condition} suggests a simple design principle for constructing shift-invariant circuits. The condition \(|\langle x|U|y\rangle|=1/\sqrt{|G|}\) requires flat amplitude magnitudes, which are naturally generated by Hadamard gates. The remaining freedom lies in the relative phases \(\theta_{xy}\), which can be adjusted by controlled-phase gates ($\mathrm{CP}(\theta)=\ket{00}\bra{00}+\ket{01}\bra{01}+\ket{10}\bra{10}+e^{\ii\theta}\ket{11}\bra{11}$) without changing these magnitudes. This motivates the Hadamard-phase (HP) circuit family we shall now introduce. As illustrated in Fig.~\ref{fig:illus} (b), the \myqc-$L$ circuit is structured as an alternating sequence of Hadamard gates and phase blocks. To respect shift invariance, we enforce three rules for \myqc-$L$: (i) each qubit is acted on by a Hadamard gate exactly once; (ii) each qubit pair is connected by at most one controlled-phase gate; and (iii) each controlled-phase gate must come together with 2 Hadamards, one before and one after the phase gate. Here, ``once'' refers to the entire $L$-layer circuit rather than to a single layer.

To express the allowed circuits mathematically, we define two types of sets of qubits.
Let \(H_i\) denote the Hadamard gate applied to qubit
\(i\in\{1,2,\ldots,n\}\). For each layer \(l\), let
\(\Lambda_l\) denote the set of qubits on which
Hadamard gates are applied in that layer. The qubits that are still allowed in terms of having a Hadamard applied to them after layer $l$, i.e.\ those that have not had a Hadamard applied to them in or before layer $l$, are denoted as the set \(\Lambdacomp{l}\).

In-between the Hadamard layers, the two sets are coupled via a controlled-phase block, denoted as $\cpgate{\Lambda_l}{\Lambdacomp{l}}{\bm{\theta}_{\Lambda_l\Lambdacomp{l}}}$, which applies pairwise operations $\mathrm{CP}_{jk}(\theta_{jk})$ between each qubit $j \in \Lambda_l$ and $k \in \Lambdacomp{l}$. \rev{Such 
pairwise controlled-phase operations should, intuitively, satisfy the necessary condition in Eq.(\ref{eq:necessarycondition}), which allows for the presence of  $e^{i\theta_{xy}}$ terms. Building upon this notation, we formally define the circuit architecture:}
\begin{definition}[\myqc-$L$ circuit]
    A quantum circuit is classified as an \myqc-L circuit if it is constructed using the following layer-wise architecture
    \begin{equation}
    \label{eq:HP}
        \uqc{L}({\bm{\theta}}) = \bigotimes_{i \in \Lambda_L} H_{i}\prod_{l=L-1}^{0} \left( \cpgate{\Lambda_l}{\Lambdacomp{l}}{\bm{\theta}_{\Lambda_l\Lambdacomp{l}}} \bigotimes_{i \in \Lambda_l} H_{i} \right),
    \end{equation}
    where the controlled-phase block is given by
    \begin{equation}
        \cpgate{\Lambda_l}{\Lambdacomp{l}}{\bm{\theta}_{\Lambda_l\Lambdacomp{l}}} = \prod_{j \in \Lambdacomp{l}} \prod_{k \in \Lambda_l} \mathrm{CP}_{jk}(\theta_{jk}).
    \end{equation}
    The parameter set $\bm{\theta}_{\Lambda_l\Lambdacomp{l}}$ is the collection of $\theta_{jk}$, $j\in \Lambda_l$, $k \in \Lambdacomp{l}$. 
\end{definition}

In the following, we show that HP circuits preserve shift invariance. This implies that such invariance is not unique to the QFT, but is a feature that arises in a wider family of quantum circuits. \rev{We focus on the cyclic group $\mathbb{Z}_{2^n}=\{0,1,\ldots,2^n-1\}$, equipped with addition modulo $2^n$, because it is convenient by virtue of being naturally represented by an $n$-qubit computational basis.}

\begin{theorem}
\label{thm:hp_circuits_are_sinv}
    \myqc-$L$ circuits are shift-invariant for $\mathbb{Z}_{2^{n}}$.
\end{theorem}

The proof derives the HP-L state analytically and verifies its shift invariance (Appendix~\ref{app:HPk are sinv}).

\rev{
The shift invariance of HP circuits can be compared to that of the QFT. For HP circuits Theorem~\ref{thm:hp_circuits_are_sinv} establishes exact shift invariance for the case of $\mathbb{Z}_{2^n}$. Moreover we give numerical evidence of approximate shift invariance more generally (Appendix H). The QFT is analytically proven to have exact shift invariance for the case of $\mathbb{Z}_{2^n}$, and analytically proven to have approximate shift invariance more generally~\cite{Shor1997} (see also Appendix H). In the approximate regime, the QFT appears to have better shift invariance. As shall be shown numerically, the HP circuit's invariance is nevertheless sufficient for factoring in the instances investigated. }

\section{\rev{\myqc-$L$ for HSP over $\mathbb{Z}$: information on generator is retained}}
\rev{Apart from shift invariance, we further demand that information on the generator is retained in the statistics of the final state in the computational basis. }

\rev{A (too) simple case of the HP-circuits illustrates the need for the second condition.} For HSP over $\mathbb Z_{2^n}$, \myqc-0, which reduces to a parallel Hadamard layer,
gives an analytically solvable procedure for the hidden subgroup
generator $2^p$. For \myqc-0, we can determine $p$ by identifying which output bits are fixed at 0. \rev{(The detailed algorithm is given in Appendix~\ref{app:transversal_hadamard_sampling}). However, the protocol for  $\myqc$-0 fails for general cyclic groups $\mathbb{Z}_q$ because the measurement distributions corresponding to different subgroups $V$ are not distinguishable. That absence of distinguishability illustrates that shift invariance, while necessary for hidden subgroup recovery, is not sufficient. An additional distinguishability condition is required.}

\begin{figure}[t]
    \includegraphics[width=0.45  \textwidth]{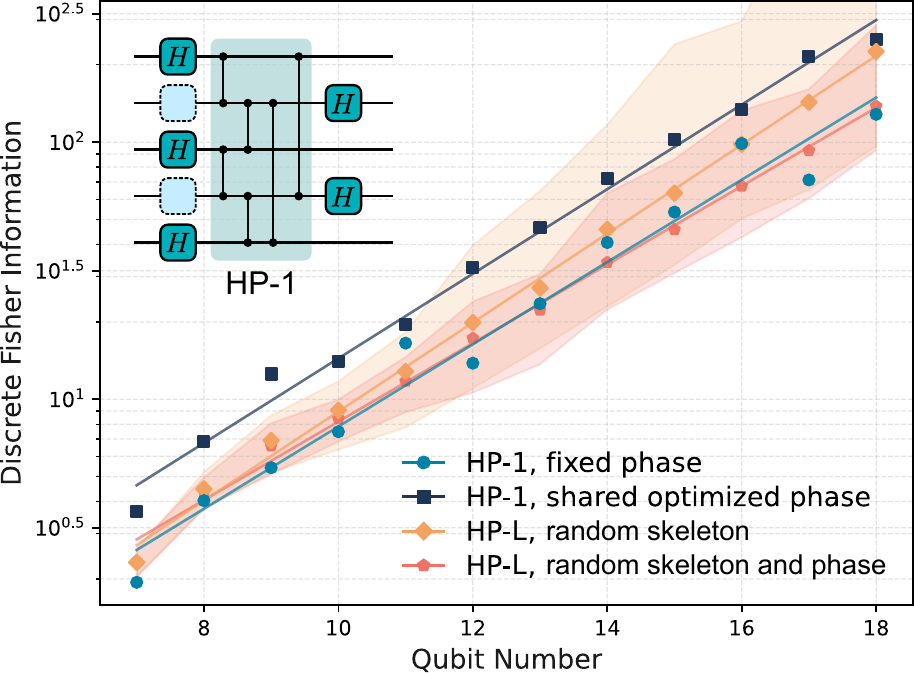}
    \caption{\label{fig:fi_scaling} 
    \textbf{\rev{Average discrete Fisher information (DFI) increasing with the number of qubits.}}\rev{Results are shown for several HP-family circuit settings. $L$ denotes the number of phase blocks layers with HP-1 having a single phase block layer.  Solid curves show the mean, with shaded bands indicating one standard deviation associated with random skeleton or phase choices.}}
\end{figure}

\rev{We use the discrete Fisher information (DFI) as a metric for distinguishability. We use DFI to quantify the impact of the parameter $r$, the subgroup generator, on the measurement outcome probabilities $\text{Pr}(x |r)$. For the \myqc~circuit/QFT denoted as $U$, $\mathrm{Pr}(x|r) = \left|\sum_q \langle x|U|qr\rangle\right|^2$. The DFI is then defined as}
\begin{equation}
\label{eq:DFI}
    \fii(r, n) = \sum_{x = 0}^{2^n -1} \frac{\left(\text{Pr}(x|r+1) - \text{Pr}(x|r)\right)^2}{\text{Pr}(x|r)},
\end{equation}

\rev{The DFI (and the continuous FI) are commonly used to quantify the distinguishability, see e.g. Ref.~\cite{weimar2025fisher}. There is a lower bound on the variance as a function of the number of samples and the DFI~\cite{chapman1951minimum, hammersley1950estimating}:}
 \begin{equation}
     \label{eq:m_bound}
     m \ge \frac{\ln\left( 1 + \epsilon^{-2} \right)}{\ln\left( 1 + \fii(r) \right)},
 \end{equation}
where $m$ is the required number of samples and \(\epsilon\) denotes the maximum tolerable standard deviation of the estimator of the parameter. \rev{In general, it is not guaranteed that the information-theoretic sample-complexity bound in Eq.~\ref{eq:m_bound} can be attained by a computationally efficient estimator.} 

\rev{The performance of our approach is similar to that of maximum likelihood in the tested regime. The maximum-likelihood approach to estimating the parameter is termed {\em sample}-efficient in that it saturates the bound on the variance for the continuous case~\cite{Paris2009, Kwon2026} and minimizes the worst-case error in the discrete case~\cite{choirat2012estimation}. We find numerically that the HP-1 circuit with DeepSets post-processing (which will be explained later) has a similar performance to the maximum likelihood in determining $r$ (appendix G). This is even though our post-processing approach is by design computationally efficient. Thus we believe that the bound of Eq.~\ref{eq:m_bound} is highly relevant to the actual scaling of the HP-1+DeepSets approach.}

For the QFT (Appendix~\ref{app:qft_fisher_information}), we can calculate the DFI as
\begin{equation}
\label{eq:qft dfi}
    \fii_{\mathrm{QFT}}(r,n) \approx \frac{4\pi^2}{9}\left(\frac{2^{2n}}{r^2}-1\right) \sim \mathcal{O}\left(2^{2n}\right).
\end{equation} 
\rev{The factor $2^{2n}/r^2$ in Eq.~\eqref{eq:qft dfi} is one way of understanding why the HSP with the QFT requires $1 \ll r \ll 2^n$, a restriction we shall also impose here.}  





\rev{For fixed-phase \myqc-1, we use the minimized DFI
\begin{equation}
    \fii_{\min}(n)
    :=
    \min_{2\le r<2^{n/\kappa}}\fii(r,n)
\label{eq:hp1_min_dfi_definition_maintext}
\end{equation}
to indicate the worst-case sensitivity of hidden subgroup generators. We find, in the instances investigated, that the $\fii_{\min}$ obeys a similar exponential growth:}

\rev{\begin{conjecture}
\label{thm:hp1_first_order_fisher_transfer_maintext}
For fixed-phase \myqc-1 and \(\kappa=4\), corresponding to the
restricted period \(2\le r<2^{n/4}\ll2^n\), the minimum DFI follows
\begin{equation}
    \fii_{\min}(n)
    \gtrsim
    \exp(kn+b),
    \label{eq:hp1_restricted_window_ln_linear_bound_maintext}
\end{equation}
for constants \(k>0\) and \(b\).
\end{conjecture}}


\rev{Conjecture~\ref{thm:hp1_first_order_fisher_transfer_maintext}
has significant analytical and numerical support.  The conjecture can be shown analytically to follow from the scaling of a particular set of dominant terms in the DFI sum (Appendix E). The scaling expression for those terms has been verified by brute force calculation over many small systems and rare-event simulations, including selected cases up to \(n=200\). Details of the derivation and supporting numerical results are provided in Appendix~\ref{app:qft_fisher_information}. The apparent exponential growth of \(\fii_{\min}\) indicates increasing distinguishability and suggests favorable sample-complexity scaling for \myqc~circuits.}

\rev{For a broader range of \myqc{} circuits, beyond \myqc{}-1,  including different phase choices and numbers of phase-block layers, the DFI is found numerically to grow exponentially.} 
Fig.~\ref{fig:fi_scaling} illustrates the minimum DFI against the qubit number $n$. Fig.~\ref{fig:fi_scaling} also compares the scaling of different \myqc~circuit configurations, where solid curves and shaded bands represent the mean and standard deviation of $\fii_{\min}(n)$. The fixed-phase \myqc-1 circuit utilizes the skeleton, $\Lambda_0 = \{1,3,5, \cdots\}$ and $\Lambda_1 = \{2,4,6, \cdots\}$, with controlled-phase gates with phase $\frac{\pi}{2^{|i-j|}}$ between qubits $i$ and $j$. The optimized \myqc-1 circuit adopts the same skeleton but treats the gate phases as variational parameters, trained to maximize the $\fii_{\min}$. The deeper \myqc{}-$L$ circuits are constructed by generating valid random skeletons $\Lambda$. \rev{The results show exponential growth within the regimes tested.}


\begin{figure}[!t]
    \centering
    \makebox[\columnwidth][l]{\textbf{(a)}}\\[-0.2em]
    \includegraphics[width=\columnwidth]{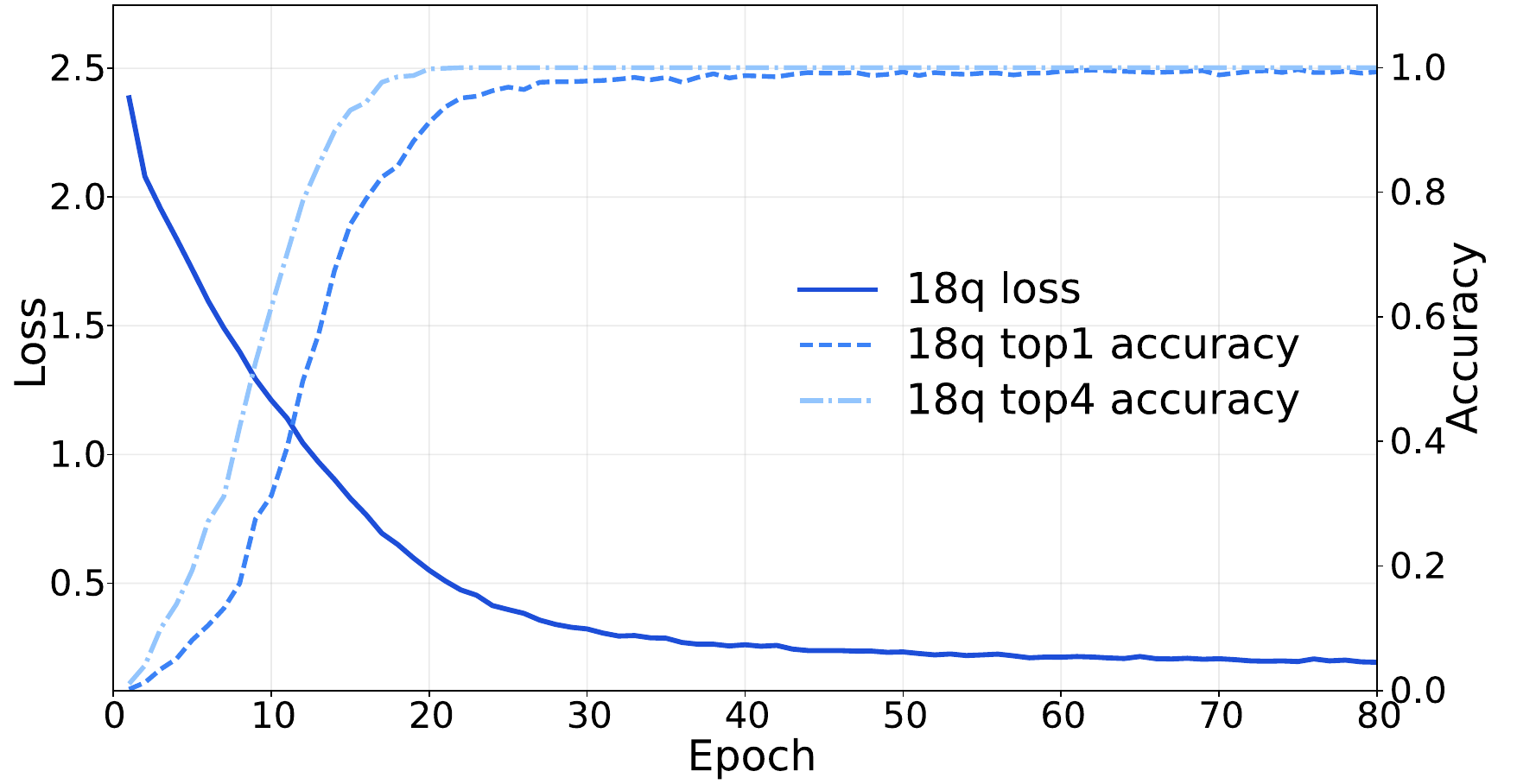}
    \vspace{0.5em}
    \makebox[\columnwidth][l]{\textbf{(b)}}\\[-0.2em]
    \includegraphics[width=\columnwidth]{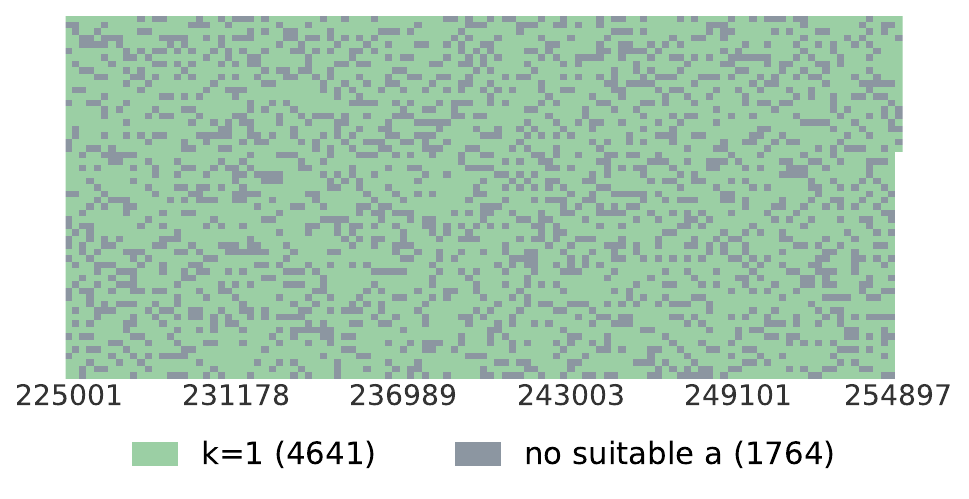}
    \caption{\label{fig:recovery_pipeline}
    \textbf{\rev{Numerical simulation shows successful Shor factoring via \myqc-1.}}
    \rev{(a) Validation loss, top-1 accuracy, and top-4 accuracy for neural-network period recovery on an 18-qubit \myqc{1} circuit.
    (b) Shor factorization results for semiprime numbers between 225001 and 254999. Green squares indicate cases that are successfully factored. Grey square indicate cases where no valid $a$ yields a detectable period.}
}
\end{figure}

\section{Efficient neural network for classical post-processing}To avoid the exponential costs associated with processing full probability distributions, \rev{we employ a neural network with $\text{poly}(n)$ parameters} to predict the hidden subgroup generator directly from the raw measurement outcomes $\{b_i\}$. The input data are generated using the optimized \myqc-1 circuit evaluated in Fig.~\ref{fig:fi_scaling}. 

\rev{We employ a Deep Sets neural network suitable for permutation invariant inputs, as the order of the samples should not influence the outcome~\cite{zaheer2017deep, edwards2016towards}. The outcomes are fed into a shared feature map \(\phi\), which maps each bitstring \(b_i\) to a feature vector. }
The feature vectors are then summed to form an aggregated, permutation-invariant representation $\sum_i\phi(b_i)$. \rev{Subsequently, a beam-search decoder \cite{sutskever2014sequence, wiseman2016sequence} 
returns a ranked list of candidate generators.
We use a polynomial number of samples to train the neural net.} This architecture avoids exponential sample-space computations. 
\rev{Further architectural and training details are provided in Appendix~\ref{sec:AppendixDecoder}.}

\section{Replacing \texorpdfstring{$\mathcal{O}(n^2)$}{O(n2)} QFT with \texorpdfstring{$\mathcal{O}(n)$}{O(n)} \myqc{}-1 in Shor's Algorithm} We evaluate the optimized \myqc{}-1 circuit in classical numerical simulations of Shor's algorithm for factoring with neural net classical post-processing. \rev{Whilst the \myqc{}-1 circuit and neural-network decoder are developed as a general framework for hidden-generator recovery in the HSP, we focus on Shor's algorithm as a representative application.} Shor's algorithm seeks the generator $r$ satisfying $a^r \equiv 1 \pmod{N}$. 
We select the base $a$ such that $1 \le r \le 2^{n/2}$, ensuring resolvability by the $n$-qubit register.
\rev{We consider factoring 18-qubit semiprime instances with \(N=225001,\ldots,254999\).}

The neural net trained successfully. \rev{As shown in Fig.~\ref{fig:recovery_pipeline}(a), the neural network achieves rapid convergence when trained on an 18-qubit instance. Both the top-1 (correct period ranked first) and top-4 accuracies increase swiftly within a few epochs.}

The combined HP-1 and neural net protocol recovers the subgroup generator in the tested instances. As illustrated in Fig.~\ref{fig:recovery_pipeline}(b), each composite integer is colored by the minimum number of guesses $k$. \rev{All solvable instances are solved on the first attempt ($k=1$). We further add noise during the testing to generically simulate experimental scenarios (see Appendix~\ref{app:noise_robustness_details} for numbers).}

\section{Summary and outlook}We propose a hardware-efficient framework for solving the HSP by shallow, shift-invariant \myqc{}-1 circuits with a polynomial-scaling neural network. To verify the scalability of this protocol, we analyzed the growth model of the DFI. Finally, we evaluated our protocol on Shor's algorithm. The numerical simulations demonstrate that our approach successfully finds the hidden subgroup with high accuracy and remains robust against moderate depolarizing noise.


Apart from the $\mathcal{O}(n)$ scaling being hardware friendly, the reduction to grouped phase gates in \myqc~circuits makes the phase operations easier to implement than individual phase gates, \rev{since} consecutive phase operations can be merged into single continuous pulses across multiple experimental platforms, including in linear optics \cite{carolan2015universal, Kok2007}, neutral-atom systems \cite{saffman2010quantum, evered2023high, mohan2025parametrized}, trapped ions \cite{Leibfried2003, Tan2013}, and NMR \cite{vandersypen2004nmr}.


\section{Acknowledgements}We acknowledge support from The Guangdong Provincial Quantum Science Strategic Initiative (Grant No. GDZX2503001) and the City University of Hong Kong (Proj. 9610623).
\section{Code Availability}All code relevant to this work is publicly available in the GitHub repository: \url{https://github.com/Fragecity/altqft}.

\bibliography{ref.bib}

@article{liu2024information,
  title={Information compression via hidden subgroup quantum autoencoders},
  author={Liu, Feiyang and Bian, Kaiming and Meng, Fei and Zhang, Wen and Dahlsten, Oscar},
  journal={npj Quantum Information},
  volume={10},
  number={1},
  pages={74},
  year={2024},
  publisher={Nature Publishing Group UK London}
}

@misc{Kitaev1995,
  author = {A. Yu. Kitaev},
  title = {Quantum Measurements and the Abelian Stabilizer Problem},
  year = {1995},
  eprint = {quant-ph/9511026},
  archivePrefix = {arXiv},
  primaryClass = {quant-ph}
}

@article{EttingerHoyerKnill2004,
  author = {Ettinger, Mark and H{\o}yer, Peter and Knill, Emanuel},
  title = {The Quantum Query Complexity of the Hidden Subgroup Problem Is Polynomial},
  journal = {Information Processing Letters},
  volume = {91},
  number = {1},
  pages = {43--48},
  year = {2004},
  doi = {10.1016/j.ipl.2004.03.008}
}

@misc{Lomont2004,
  author = {Chris Lomont},
  title = {The Hidden Subgroup Problem: Review and Open Problems},
  year = {2004},
  eprint = {quant-ph/0411037},
  archivePrefix = {arXiv},
  primaryClass = {quant-ph}
}

@article{Jozsa2001,
  author = {Richard Jozsa},
  title = {Quantum Factoring, Discrete Logarithms, and the Hidden Subgroup Problem},
  journal = {Computing in Science \& Engineering},
  volume = {3},
  number = {2},
  pages = {34--43},
  year = {2001},
  doi = {10.1109/5992.909005}
}

@ARTICLE{Kwon2026,
       author = {{Kwon}, Hyukgun and {Lie}, Seok Hyung and {Jiang}, Liang},
        title = "{Universal Sample Complexity Bounds in Quantum Learning Theory via Fisher Information Matrix}",
      journal = {arXiv e-prints},
         year = 2026,
        month = feb,
          eid = {arXiv:2602.21510},
        pages = {arXiv:2602.21510},
          doi = {10.48550/arXiv.2602.21510},
archivePrefix = {arXiv},
       eprint = {2602.21510},
 primaryClass = {quant-ph},
       adsurl = {https://ui.adsabs.harvard.edu/abs/2026arXiv260221510K}
}

@ARTICLE{Paris2009,
       author = {{Paris}, Matteo G.~A.},
        title = "{Quantum estimation for quantum technology}",
      journal = {arXiv e-prints},
         year = 2008,
        month = apr,
          eid = {arXiv:0804.2981},
        pages = {arXiv:0804.2981},
          doi = {10.48550/arXiv.0804.2981},
archivePrefix = {arXiv},
       eprint = {0804.2981},
 primaryClass = {quant-ph},
       adsurl = {https://ui.adsabs.harvard.edu/abs/2008arXiv0804.2981P}
}

@INPROCEEDINGS{CleveWatrous2000,
  author={Cleve, R. and Watrous, J.},
  booktitle={Proceedings 41st Annual Symposium on Foundations of Computer Science}, 
  title={Fast parallel circuits for the quantum Fourier transform}, 
  year={2000},
  volume={},
  number={},
  pages={526-536},
  doi={10.1109/SFCS.2000.892140}}

@INPROCEEDINGS{HalesHallgren2000,
  author={Hales, L. and Hallgren, S.},
  booktitle={Proceedings 41st Annual Symposium on Foundations of Computer Science}, 
  title={An improved quantum Fourier transform algorithm and applications}, 
  year={2000},
  volume={},
  number={},
  pages={515-525},
  doi={10.1109/SFCS.2000.892139}}

@book{NielsenChuang2000,
  author = {Michael A. Nielsen and Isaac L. Chuang},
  title = {Quantum Computation and Quantum Information},
  publisher = {Cambridge University Press},
  year = {2000}
}

@article{Zuchongzhi,
  title = {Establishing a New Benchmark in Quantum Computational Advantage with 105-qubit {Zuchongzhi} 3.0 Processor},
  author = {Gao, Dongxin and Fan, Daojin and Zha, Chen and others},
  journal = {Phys. Rev. Lett.},
  volume = {134},
  issue = {9},
  pages = {090601},
  numpages = {7},
  year = {2025},
  month = {Mar},
  publisher = {American Physical Society},
  doi = {10.1103/PhysRevLett.134.090601},
  url = {https://link.aps.org/doi/10.1103/PhysRevLett.134.090601}
}

@article{Simon1997,
  author = {Daniel R. Simon},
  title = {On the Power of Quantum Computation},
  journal = {SIAM Journal on Computing},
  volume = {26},
  number = {5},
  pages = {1474--1483},
  year = {1997},
  doi = {10.1137/S0097539796298637}
}

@inproceedings{Shor1994,
  author = {Peter W. Shor},
  title = {Algorithms for Quantum Computation: Discrete Logarithms and Factoring},
  booktitle = {Proceedings 35th Annual Symposium on Foundations of Computer Science},
  pages = {124--134},
  publisher = {IEEE},
  year = {1994},
  doi = {10.1109/SFCS.1994.365700}
}

@article{Shor1997,
   author = {Shor, Peter W.},
   title = {Polynomial-Time Algorithms for Prime Factorization and Discrete Logarithms on a Quantum Computer},
   journal = {SIAM Journal on Computing},
   volume = {26},
   number = {5},
   pages = {1484-1509},
   DOI = {10.1137/S0097539795293172},
   url = {https://doi.org/10.1137/S0097539795293172 ,eprint = https://doi.org/10.1137/S0097539795293172 , abstract = A digital computer is generally believed to be an efficient universal computing device; that is, it is believed able to simulate any physical computing device with an increase in computation time by at most a polynomial factor. This may not be true when quantum mechanics is taken into consideration. This paper considers factoring integers and finding discrete logarithms, two problems which are generally thought to be hard on a classical computer and which have been used as the basis of several proposed cryptosystems. Efficient randomized algorithms are given for these two problems on a hypothetical quantum computer. These algorithms take a number of steps polynomial in the input size, e.g., the number of digits of the integer to be factored.},
   year = {1997},
   type = {Journal Article}
}

@article{kim2023evidence,
  title={Evidence for the utility of quantum computing before fault tolerance},
  author={Kim, Youngseok and Eddins, Andrew and Anand, Sajant and Wei, Ken Xuan and Van Den Berg, Ewout and Rosenblatt, Sami and Nayfeh, Hasan and Wu, Yantao and Zaletel, Michael and Temme, Kristan and others},
  journal={Nature},
  volume={618},
  number={7965},
  pages={500--505},
  year={2023},
  publisher={Nature Publishing Group UK London}
}

@article{google2025quantum,
  author  = {{Google Quantum AI and Collaborators}},
  title   = {Quantum error correction below the surface code threshold},
  journal = {Nature},
  volume  = {638},
  number  = {8052},
  pages   = {920--926},
  year    = {2025},
  doi     = {10.1038/s41586-024-08449-y}
}

@inproceedings{tang2021cutqc,
  title={{CutQC}: using small quantum computers for large quantum circuit evaluations},
  author={Tang, Wei and others},
  booktitle={Proceedings of the 26th ACM International Conference on Architectural Support for Programming Languages and Operating Systems},
  pages={473--486},
  year={2021}
}

@article{Peng2020,
  title = {Simulating Large Quantum Circuits on a Small Quantum Computer},
  author = {Peng, Tianyi and Harrow, Aram W. and Ozols, Maris and Wu, Xiaodi},
  journal = {Phys. Rev. Lett.},
  volume = {125},
  issue = {15},
  pages = {150504},
  numpages = {6},
  year = {2020},
  month = {Oct},
  publisher = {American Physical Society},
  doi = {10.1103/PhysRevLett.125.150504},
  url = {https://link.aps.org/doi/10.1103/PhysRevLett.125.150504}
}

@misc{Coppersmith2002,
  author = {{Coppersmith}, D.},
  title = {An approximate Fourier transform useful in quantum factoring},
  year = {2002},
  eprint = {quant-ph/0201067},
  archivePrefix = {arXiv},
  primaryClass = {quant-ph}
}

@article{park2023reducing,
  title={Reducing CNOT count in quantum Fourier transform for the linear nearest-neighbor architecture},
  author={Park, Byeongyong and Ahn, Doyeol},
  journal={Scientific Reports},
  volume={13},
  number={1},
  pages={8638},
  year={2023},
  publisher={Nature Publishing Group UK London}
}

@article{baumer2024quantum,
  title={Quantum Fourier transform using dynamic circuits},
  author={B{\"a}umer, Elisa and others},
  journal={arXiv preprint arXiv:2403.09514},
  year={2024}
}

@article{aumann2026demonstrating,
  title={Demonstrating Record Fidelity for the Quantum Fourier Transform},
  author={Aumann, Philipp and Fellner, Michael and Alber, David and Cykiert, Max and Fleckenstein, Christoph and ter Hoeven, Roeland and Stenzel, Leo and Valencia-Tortora, Riccardo J and Lechner, Wolfgang},
  journal={arXiv preprint arXiv:2604.12465},
  year={2026}
}

@article{saffman2010quantum,
  title={Quantum information with {Rydberg} atoms},
  author={Saffman, Mark and Walker, Thad G and M{\o}lmer, Klaus},
  journal={Reviews of modern physics},
  volume={82},
  number={3},
  pages={2313--2363},
  year={2010},
  publisher={APS}
}

@article{Kok2007,
  title = {Linear optical quantum computing with photonic qubits},
  author = {Kok, Pieter and Munro, W. J. and Nemoto, Kae and Ralph, T. C. and Dowling, Jonathan P. and Milburn, G. J.},
  journal = {Rev. Mod. Phys.},
  volume = {79},
  issue = {1},
  pages = {135--174},
  numpages = {0},
  year = {2007},
  month = {Jan},
  publisher = {American Physical Society},
  doi = {10.1103/RevModPhys.79.135},
  url = {https://link.aps.org/doi/10.1103/RevModPhys.79.135}
}

@article{evered2023high,
  title={High-fidelity parallel entangling gates on a neutral-atom quantum computer},
  author={Evered, Simon J and Bluvstein, Dolev and Kalinowski, Marcin and Ebadi, Sepehr and Manovitz, Tom and Zhou, Hengyun and Li, Sophie H and Geim, Alexandra A and Wang, Tout T and Maskara, Nishad and others},
  journal={Nature},
  volume={622},
  number={7982},
  pages={268--272},
  year={2023},
  publisher={Nature Publishing Group UK London}
}

@article{mohan2025parametrized,
  title={Parametrized multiqubit gates for neutral-atom quantum platforms},
  author={Mohan, Madhav and De Hond, Julius and Kokkelmans, Servaas},
  journal={Physical Review Applied},
  volume={23},
  number={5},
  pages={054074},
  year={2025},
  publisher={APS}
}

@article{carolan2015universal,
  title={Universal linear optics},
  author={Carolan, Jacques and Harrold, Christopher and Sparrow, Chris and Mart{\'\i}n-L{\'o}pez, Enrique and Russell, Nicholas J and Silverstone, Joshua W and Shadbolt, Peter J and Matsuda, Nobuyuki and Oguma, Manabu and Itoh, Mikitaka and others},
  journal={Science},
  volume={349},
  number={6249},
  pages={711--716},
  year={2015},
  publisher={American Association for the Advancement of Science}
}

@article{Leibfried2003,
  author  = {Leibfried, D. and DeMarco, B. and Meyer, V. and Lucas, D. and Barrett, M. and Britton, J. and Itano, W. M. and Jelenkovi{\'c}, B. and Langer, C. and Rosenband, T. and Wineland, D. J.},
  title   = {Experimental demonstration of a robust, high-fidelity geometric two ion-qubit phase gate},
  journal = {Nature},
  volume  = {422},
  number  = {6930},
  pages   = {412--415},
  year    = {2003},
  doi     = {10.1038/nature01492},
  url     = {https://doi.org/10.1038/nature01492}
}

@article{Tan2013,
  title = {Demonstration of a Dressed-State Phase Gate for Trapped Ions},
  author = {Tan, T. R. and Gaebler, J. P. and Bowler, R. and Lin, Y. and Jost, J. D. and Leibfried, D. and Wineland, D. J.},
  journal = {Phys. Rev. Lett.},
  volume = {110},
  issue = {26},
  pages = {263002},
  numpages = {5},
  year = {2013},
  month = {Jun},
  publisher = {American Physical Society},
  doi = {10.1103/PhysRevLett.110.263002},
  url = {https://link.aps.org/doi/10.1103/PhysRevLett.110.263002}
}

@article{vandersypen2004nmr,
  title={{NMR} techniques for quantum control and computation},
  author={Vandersypen, Lieven MK and Chuang, Isaac L},
  journal={Reviews of Modern Physics},
  volume={76},
  number={4},
  pages={1037--1069},
  year={2004},
  publisher={APS}
}

@article{chapman1951minimum,
  title={Minimum variance estimation without regularity assumptions},
  author={Chapman, Douglas G and Robbins, Herbert},
  journal={The Annals of Mathematical Statistics},
  pages={581--586},
  year={1951},
  publisher={JSTOR}
}

@article{hammersley1950estimating,
  title={On estimating restricted parameters},
  author={Hammersley, John M},
  journal={Journal of the Royal Statistical Society. Series B (Methodological)},
  volume={12},
  number={2},
  pages={192--240},
  year={1950},
  publisher={JSTOR}
}

@article{zaheer2017deep,
  title={Deep sets},
  author={Zaheer, Manzil and Kottur, Satwik and Ravanbakhsh, Siamak and Poczos, Barnabas and Salakhutdinov, Russ R and Smola, Alexander J},
  journal={Advances in neural information processing systems},
  volume={30},
  year={2017}
}

@article{edwards2016towards,
  title={Towards a neural statistician},
  author={Edwards, Harrison and Storkey, Amos},
  journal={arXiv preprint arXiv:1606.02185},
  year={2016}
}

@inproceedings{sutskever2014sequence,
  title={Sequence to sequence learning with neural networks},
  author={Sutskever, Ilya and Vinyals, Oriol and Le, Quoc V},
  booktitle={Advances in neural information processing systems},
  volume={27},
  year={2014}
}

@inproceedings{wiseman2016sequence,
  title={Sequence-to-Sequence Learning as Beam-Search Optimization},
  author={Wiseman, Sam and Rush, Alexander M},
  booktitle={Proceedings of the 2016 Conference on Empirical Methods in Natural Language Processing},
  pages={1296--1306},
  year={2016}
}

@article{takagi2022fundamental,
  title={Fundamental limits of quantum error mitigation},
  author={Takagi, Ryuji and Endo, Suguru and Minagawa, Shintaro and Gu, Mile},
  journal={npj Quantum Information},
  volume={8},
  number={1},
  pages={114},
  year={2022},
  publisher={Nature Publishing Group UK London}
}

@article{takagi2023universal,
  title={Universal sampling lower bounds for quantum error mitigation},
  author={Takagi, Ryuji and Tajima, Hiroyasu and Gu, Mile},
  journal={Physical Review Letters},
  volume={131},
  number={21},
  pages={210602},
  year={2023},
  publisher={APS}
}

@article{weimar2025fisher,
  title={Fisher information flow in artificial neural networks},
  author={Weimar, Maximilian and Rachbauer, Lukas M and Starshynov, Ilya and Faccio, Daniele and Adilova, Linara and Bouchet, Dorian and Rotter, Stefan},
  journal={Physical Review X},
  volume={15},
  number={3},
  pages={031072},
  year={2025},
  publisher={APS}
}

@article{vandenBerg2023,
  author  = {van den Berg, Ewout and Minev, Zlatko K. and Kandala, Abhinav and Temme, Kristan},
  title   = {Probabilistic error cancellation with sparse Pauli--Lindblad models on noisy quantum processors},
  journal = {Nature Physics},
  volume  = {19},
  pages   = {1116--1121},
  year    = {2023},
  doi     = {10.1038/s41567-023-02042-2}
}

@article{GoogleQAI2023,
  author  = {{Google Quantum AI}},
  title   = {Suppressing quantum errors by scaling a surface code logical qubit},
  journal = {Nature},
  volume  = {614},
  pages   = {676--681},
  year    = {2023},
  doi     = {10.1038/s41586-022-05434-1}
}

@article{Bausch2024,
  author  = {Bausch, Johannes and Senior, Andrew W. and Heras, Francisco J. H. and others},
  title   = {Learning high-accuracy error decoding for quantum processors},
  journal = {Nature},
  volume  = {635},
  pages   = {834--840},
  year    = {2024},
  doi     = {10.1038/s41586-024-08148-8}
}

@article{Will2025,
  author  = {Will, M. and Cochran, T. A. and Rosenberg, E. and others},
  title   = {Probing non-equilibrium topological order on a quantum processor},
  journal = {Nature},
  volume  = {645},
  pages   = {348--353},
  year    = {2025},
  doi     = {10.1038/s41586-025-09456-3}
}

@article{Evered2023,
  author  = {Evered, Simon J. and Bluvstein, Dolev and Kalinowski, Marcin
             and Ebadi, Sepehr and Manovitz, Tom and Zhou, Hengyun
             and Li, Sophie H. and Geim, Alexandra A. and Wang, Tout T.
             and Maskara, Nishad and Levine, Harry and Semeghini, Giulia
             and Greiner, Markus and Vuleti{\'c}, Vladan and Lukin, Mikhail D.},
  title   = {High-fidelity parallel entangling gates on a neutral-atom quantum computer},
  journal = {Nature},
  volume  = {622},
  pages   = {268--272},
  year    = {2023},
  doi     = {10.1038/s41586-023-06481-y}
}

@article{GoogleQAI2025,
  author  = {{Google Quantum AI and Collaborators}},
  title   = {Quantum error correction below the surface code threshold},
  journal = {Nature},
  volume  = {638},
  pages   = {920--926},
  year    = {2025},
  doi     = {10.1038/s41586-024-08449-y}
}

@article{choirat2012estimation,
  title={Estimation in discrete parameter models},
  author={Choirat, Christine and Seri, Raffaello},
  journal={Statistical Science},
  volume={27},
  number={2},
  pages={278--293},
  year={2012},
  month={May},
  publisher={Institute of Mathematical Statistics},
  doi={10.1214/11-STS371},
  url={https://doi.org/10.1214/11-STS371}
}
\newpage
\onecolumngrid
\allowdisplaybreaks
\let\section\originalsection
\appendix
\setcounter{secnumdepth}{1}

\newtheorem{lemma}{Lemma}
\newenvironment{theoremrestatement}[2]{%
    \begin{trivlist}
    \item[\hskip\labelsep\bfseries Theorem~\ref{#1} (#2).]\itshape
}{%
    \end{trivlist}
}

\section*{Appendix}
\makeatletter
\def\@sectioncntformat#1{\csname the#1\endcsname.\quad}
\def\@hangfrom@section#1#2#3{#1#2#3}
\makeatother

\section{Hidden Subgroup Problem and Fourier Sampling}
\label{append: hsp and fs}

For completeness, this section reviews the standard quantum algorithm for solving the Hidden Subgroup Problem (HSP) over finite abelian groups \cite{Kitaev1995,Jozsa2001,Lomont2004}. Formally, the HSP is defined as follows:

\begin{definition}[Hidden Subgroup Problem]
    Given a function $f: G \to X$ defined on a group $G$, with the promise that there exists a hidden subgroup $V \leq G$ such that for all $g_1, g_2 \in G$
    \begin{equation}
        f(g_1) = f(g_2) \Longleftrightarrow g_1 - g_2 \in V
    \end{equation}
    The objective of this task is to identify the hidden subgroup $V$.
\end{definition}

In the context of this paper, we restrict our attention to finite abelian groups. For such groups, the fundamental theorem of finite abelian groups allows us to express the group structure as a direct product of cyclic groups:
\begin{equation}
    G \cong \mathbb{Z}_{p_1} \times \mathbb{Z}_{p_2} \times \cdots \times \mathbb{Z}_{p_k}.
\end{equation}

The standard quantum procedure for solving the hidden subgroup problem begins by preparing a uniform superposition over all elements of the group $G$ \cite{Kitaev1995,NielsenChuang2000}. We initialize the system in the state
\begin{equation}
    \ket{\psi} = \frac{1}{\sqrt{|G|}} \sum_{g \in G} \ket{g} \ket{0},
\end{equation}
where $|G|$ is the number of elements of the group $G$. The first register holds the group elements, and the second register is initialized to a standard basis state. Next, we evaluate the quantum oracle corresponding to the function $f$, which acts on the two registers and entangles the group elements with their corresponding function values
\begin{equation}
    \ket{\psi} = \frac{1}{\sqrt{|G|}} \sum_{g \in G} \ket{g} \ket{f(g)}.
\end{equation}


Grouping the superposition terms according to cosets is illuminating, as the function $f$ is constant within each coset of the hidden subgroup $V$ and distinct across different cosets.

Before proceeding, it is useful to formalize the state notation for a subset. For any subset $S \subseteq G$, we define its corresponding normalized quantum state $\ket{S}$ as the uniform superposition over all elements within the set:
\begin{equation}
    \ket{S} = \frac{1}{\sqrt{|S|}} \sum_{s \in S} \ket{s}.
\end{equation}
Using this definition, the normalized state for a specific coset $c+V$ is directly given by
\begin{equation}
    \ket{c+V} = \frac{1}{\sqrt{|V|}} \sum_{v \in V} \ket{c+v}.
\end{equation}

By applying this coset state definition, the post-oracle state can be elegantly rewritten as a superposition of these orthogonal coset states, each perfectly entangled with its unique function evaluation:
\begin{equation}
    \ket{\psi} = \sqrt{\frac{|V|}{|G|}} \sum_{c\in C} \ket{c+V} \ket{f(c)},
\end{equation}
where $C$ is the set of all possible coset representatives. For simplicity, we assume that the second register is measured immediately after the oracle query. Measuring the second register yields a specific outcome $f(c)$ with a uniform probability of $|V|/|G|$. Consequently, the first register collapses into the corresponding pure coset state:
\begin{equation}
    \ket{\psi_{\text{post}}} = \ket{c+V}.
\end{equation}
From this point forward, our analysis will focus exclusively on the quantum state of this first register.

Whether the second register is physically measured to collapse the first register into a pure coset state $\ket{c+V}$, or traced out to leave it in a mixed state, the resulting Fourier-basis statistics remain identical \cite{NielsenChuang2000}. The quantum Fourier transform produces a measurement distribution that is independent of the coset representative. For simplicity, we assume the second register is measured immediately after the oracle query, allowing our subsequent analysis to focus exclusively on the pure coset state:
\begin{equation}
    \ket{\psi_{\text{post}}} = \ket{c+V}.
\end{equation}

A direct measurement in the computational basis is uninformative. Indeed, using the definition of the coset state,
\begin{equation}
    \langle g|c+V\rangle
    = \frac{1}{\sqrt{|V|}} \sum_{v \in V} \langle g|c+v\rangle
    = \frac{1}{\sqrt{|V|}} \sum_{v \in V} \delta_{g,c+v}.
\end{equation}
For fixed $c$, the map $v \mapsto c+v$ is a bijection from $V$ to the coset $c+V$. Therefore the sum $\sum_{v \in V} \delta_{g,c+v}$ equals $1$ when $g \in c+V$ and equals $0$ when $g \notin c+V$. Hence
\begin{equation}
    \Pr[g\,|\,c,V] = |\langle g|c+V\rangle|^2 =
    \begin{cases}
        1/|V|,& g\in c+V,\\
        0,& g\notin c+V,
    \end{cases}
\end{equation}
so the observed element is uniform on the unknown coset. Averaging over uniformly random cosets gives the marginal distribution
\begin{equation}
    \Pr[g]=\frac{1}{|G|},
\end{equation}
which contains no direct information about the hidden subgroup $V$.

To extract meaningful information, we must apply the QFT over the group $G$ before measuring \cite{Kitaev1995,Jozsa2001}. Let $\chi_y(g)$ denote the group characters. The QFT maps a computational basis state $\ket{g}$ as follows:
\begin{equation}
    \text{QFT} \ket{g} = \frac{1}{\sqrt{|G|}} \sum_{y \in G} \chi_y(g) \ket{y}.
\end{equation}
Applying this transformation to our pure coset state $\ket{c+V}$ yields:
\begin{align}
    \text{QFT} \ket{c+V} &= \frac{1}{\sqrt{|V|}} \sum_{h \in V} \text{QFT} \ket{c+v} \nonumber \\
    &= \frac{1}{\sqrt{|V||G|}} \sum_{v \in V} \sum_{y \in G} \chi_y(c+v) \ket{y} \nonumber \\
    &= \frac{1}{\sqrt{|V||G|}} \sum_{y \in G} \chi_y(c) \left( \sum_{v \in V} \chi_y(v) \right) \ket{y}.
\end{align}
By the orthogonality of group characters, the sum $\sum_{v \in V} \chi_y(v)$ evaluates to $|V|$ if the character $\chi_y$ is trivial on $V$ (i.e., $y$ belongs to the orthogonal subgroup $V^\perp$), and evaluates to $0$ otherwise. Thus, the state simplifies to a superposition over $V^\perp$:
\begin{equation}
    \text{QFT} \ket{c+V} = \sqrt{\frac{|V|}{|G|}} \sum_{y \in V^\perp} \chi_y(c) \ket{y}.
\end{equation}
Measuring this state in the computational basis gives the probability of observing an outcome $y \in V^\perp$:
\begin{equation}
    \Pr[y] = \left| \sqrt{\frac{|V|}{|G|}} \chi_y(c) \right|^2 = \frac{|V|}{|G|} |\chi_y(c)|^2 = \frac{|V|}{|G|}.
\end{equation}
Crucially, because $|\chi_y(c)|^2 = 1$ for all group characters, the resulting measurement distribution is uniformly distributed over $V^\perp$ and is strictly independent of the random coset representative $c$. This shift invariance is the fundamental property that allows the extraction of subgroup information by effectively removing the unknown shift at the probability level.

The QFT over $G$ is defined using the character group $\hat{G}$ by
\begin{equation}
F_G|g\rangle = \frac{1}{\sqrt{|G|}} \sum_{\chi\in\hat{G}} \chi(g)\,|\chi\rangle,
\end{equation}
where $\{|\chi\rangle:\chi\in\hat{G}\}$ is the Fourier basis and each character $\chi:G\to \mathbb{C}^\times$ is a group homomorphism. Here $\mathbb{C}^\times=\mathbb{C}\setminus\{0\}$ denotes the multiplicative group of nonzero complex numbers, so with additive notation on $G$ one has $\chi(g_1+g_2)=\chi(g_1)\chi(g_2)$. For finite abelian $G$, every character has unit modulus. Applying $F_G$ to a coset state gives
\begin{align}
F_G|c+V\rangle
&= \frac{1}{\sqrt{|V|}} \sum_{v\in V} F_G|c+v\rangle \nonumber\\
&= \frac{1}{\sqrt{|V||G|}} \sum_{\chi\in\hat{G}} \sum_{v\in V} \chi(c+v)\,|\chi\rangle \nonumber\\
&= \frac{1}{\sqrt{|V||G|}} \sum_{\chi\in\hat{G}} \chi(c)\left(\sum_{v\in V}\chi(v)\right)|\chi\rangle.
\end{align}
The character sum over $V$ satisfies
\begin{equation}
\sum_{v\in V}\chi(v)=
\begin{cases}
|V|,& \chi\in V^\perp,\\
0,& \chi\notin V^\perp,
\end{cases}
\qquad
V^\perp=\{\chi\in\hat{G}:\chi(v)=1\ \text{for all }v\in V\}.
\end{equation}
Hence
\begin{equation}
F_G|c+V\rangle
=
\sqrt{\frac{|V|}{|G|}}
\sum_{\chi\in V^\perp} \chi(c)\,|\chi\rangle.
\end{equation}
The Fourier-basis measurement distribution is therefore
\begin{equation}
\text{Pr}(\chi\,\mid\,c,V)=
\begin{cases}
|V|/|G|,& \chi\in V^\perp,\\
0,& \chi\notin V^\perp,
\end{cases}
\end{equation}
which is independent of the coset representative $c$. The same formula holds for the mixed state $\rho_V$, since it is a convex combination of pure coset states with identical Fourier statistics.

Repeating the procedure therefore produces samples from $V^\perp$. In the finite abelian setting, classical post-processing reconstructs $V^\perp$ from the sample span and then recovers the hidden subgroup using the duality identity \cite{Kitaev1995,EttingerHoyerKnill2004,Lomont2004}
\begin{equation}
(V^\perp)^\perp = V.
\end{equation}
Fourier sampling thus converts hidden translational structure in the oracle into explicit linear constraints in the dual group.

\section{Necessary condition of shift invariance}
\label{app:necessary condition of sinv}
Here we make explicit the probability-level shift-invariance condition used in Fourier sampling.
Let $G$ be a finite abelian group, and let $\{\ket{g}:g\in G\}$ denote the computational basis.
For $c\in G$, define the shift operator
\begin{equation}
    T_c\ket{g}=\ket{g+c}.
\end{equation}
For a subgroup $V\le G$, the normalized subgroup state is
\begin{equation}
    \ket{c + V}=\frac{1}{\sqrt{|V|}}\sum_{v\in V}\ket{c + v}.
\end{equation}
For HSP sampling, shift invariance requires that the measurement distribution after applying $U$ is strictly independent of the coset representative:
\begin{equation}
    \left|\bra{x}U\ket{V}\right|^2
    =
    \left|\bra{x}U T_c\ket{V}\right|^2,
    \qquad
    \forall x,c\in G,\ \forall V\le G.
    \label{eq:prob_shift_invariance}
\end{equation}
Operationally, this means $U\ket{V}$ and $U\ket{c+V}$ may differ by a relative phase, but they must yield identical probabilities in the computational basis. Any viable QFT replacement must therefore satisfy this condition.
\begin{restatedproposition}[Necessary condition of shift invariance]
    Suppose $U$ is a unitary on the group register that satisfies Eq.~\eqref{eq:prob_shift_invariance}. Then, for every $x,y\in G$,
    \begin{equation}
        \left|\bra{x}U\ket{y}\right|=\frac{1}{\sqrt{|G|}}.
    \end{equation}
    Equivalently, there exist phases $\theta_{xy}\in\mathbb{R}$ such that
    \begin{equation}
        \bra{x}U\ket{y}=\frac{1}{\sqrt{|G|}}e^{\ii\theta_{xy}}.
    \end{equation}
\end{restatedproposition}
To determine how this restricts the matrix elements of $U$, we can examine the trivial subgroup $V=\{0\}$. In this case, the subgroup state is simply the basis state $\ket{0}$, meaning the shift-invariance condition directly constrains the columns of $U$. Since the output probabilities for $\ket{0}$ and any shift $\ket{g}$ must be equal, for any fixed outcome $x$, all entries in the $x$-th row of $U$ must share the same modulus. Unitarity then forces this common modulus to be exactly $1/\sqrt{|G|}$. This leads directly to Proposition~\ref{thm:necessary_modulus_condition}.

\begin{proof}[Proof of Proposition~\ref{thm:necessary_modulus_condition}]
    Take the trivial subgroup $V=\{0\}$.
    Then $\ket{V}=\ket{0}$ and $T_c\ket{V}=\ket{c}$ for every $c\in G$.
    Equation~\eqref{eq:prob_shift_invariance} therefore implies
    \begin{equation}
        \left|\bra{x}U\ket{0}\right|^2
        =
        \left|\bra{x}U\ket{c}\right|^2,
        \qquad
        \forall x,c\in G.
    \end{equation}
    Using the matrix element convention $U_{xy}=\bra{x}U\ket{y}$, this equality shows that $|U_{x0}|^2 = |U_{xc}|^2$ for all $c \in G$.
    In other words, for each fixed row $x$, all entries have the same squared modulus.
    Denote this common value by $a_x$.
    Since $U$ is unitary, the sum of the squared moduli across any row must be normalized to 1:
    \begin{equation}
        1=\sum_{y\in G}|U_{xy}|^2=|G|  a_x.
    \end{equation}
    Hence, $a_x=1/|G|$ for every $x\in G$.
    Taking the square root yields the uniform modulus for all matrix elements:
    \begin{equation}
        |U_{xy}|=\frac{1}{\sqrt{|G|}}.
    \end{equation}
    Because this modulus is strictly positive, every entry admits a polar decomposition
    \begin{equation}
        U_{xy}=\frac{1}{\sqrt{|G|}}e^{\ii\theta_{xy}},
        \qquad \theta_{xy}\in\mathbb{R}.
    \end{equation}
    This proves the claim.
\end{proof}

\section{\myqc-$L$ circuits are shift invariant over $\mathbb{Z}_{2^{n}}$}
\label{app:HPk are sinv}
 This appendix aims to rigorously prove that the \myqc-$L$ circuit family satisfies the shift-invariance property. Establishing this theorem is of significant practical importance: it fundamentally guarantees that the output statistics obtained from this shallow circuit structure can be reliably used for the classical post-processing reconstruction of the hidden subgroup.

\begin{restatedtheorem}
The \myqc-$L$ circuit is shift-invariant over $\mathbb{Z}_{2^n}$.
\end{restatedtheorem}

The core strategy of the proof is as follows: First, we explicitly calculate the sequential effect of each layer of the \myqc-$L$ circuit on the initial basis state, deriving the analytical expression of the final quantum state. Next, we obtain the corresponding probability amplitudes by calculating the matrix elements of this output state in the computational basis. Finally, we compute the squared modulus of these amplitudes over the subgroup $V$ and its coset $c+V$ respectively, using algebraic simplification to factor out a pure phase factor related to the coset representative $c$. Because the modulus of this pure phase factor is strictly $1$, it naturally cancels out when calculating the measurement probabilities, thereby rigorously proving mathematically that the output probability distribution is completely independent of random shifts.

The $\Lambdacomp{l}$ is given by
\begin{equation}
    \Lambdacomp{l} = \{ 1,2, \cdots, n \} \backslash  \bigcup_{i \leq l} \Lambda_i ~.
\end{equation}
The $\{\Lambda_l\}$ is a partition of index set $\{1,2,\cdots n\}$, which means
\begin{equation}
    \bigcup  \Lambda_l =  \{ 1,2, \cdots, n \}, \quad \Lambda_i \cap \Lambda_j = \emptyset.
\end{equation}

\begin{proof}
    The first step is to calculate the effect of the first layer, explicitly keeping the normalization factor $\frac{1}{\sqrt{2}}$ from the Hadamard gates:
    \begin{align}
        &\cpgate{\Lambda_0}{\Lambdacomp{0}}{\bm{\theta}_{\Lambda_0 \Lambdacomp{0}}  } \bigotimes_{i \in \Lambda_0} H_{i}\ket{b_1 b_2\cdots b_n}\notag\\
        =& \cpgate{\Lambda_0}{\Lambdacomp{0}}{\bm{\theta}_{\Lambda_0\Lambdacomp{0}}} \bigotimes_{i \in \Lambda_0}\frac{1}{\sqrt{2}}(\ket{0} + e^{\pi \ii b_i} \ket{1}) \bigotimes_{k \in \bar \Lambda_0} \ket{b_k} \\
        =& \bigotimes_{i \in \Lambda_0}\frac{1}{\sqrt{2}}\left(\ket{0} + \exp\left[\pi \ii b_i + \ii \sum_{j \in \Lambdacomp{0}} \theta_{ij} b_j\right] \ket{1}\right) \bigotimes_{k \in \bar \Lambda_0} \ket{b_k}.
    \end{align}
    Note that since the qubits in $\Lambdacomp{0}$ have not yet been acted upon by Hadamard gates, they remain in the computational basis. 
    
    Since the subsets \(\Lambda_l\) form a partition of the qubits:
    \begin{equation}
        \bigcup_{l=1}^L \Lambda_l = \{1,2,\cdots, n\}, \quad
        \Lambda_i \cap \Lambda_j = \emptyset,
    \end{equation}
    and the subsets \(\bar{\Lambda}_l\) are defined as the remaining qubits that have not yet been acted on by Hadamard gates:
    \begin{equation}
        \Lambdacomp{l} = \{1,2,\cdots, n\}\setminus \left(\bigcup_{k = 1}^l \Lambda_k\right),
    \end{equation}
    the sequential action of all $L$ layers ensures every qubit $i$ is acted on by a Hadamard gate exactly once. Let $l(i)$ denote the specific stage where $i \in \Lambda_{l(i)}$. The final state after the full circuit becomes:
    \begin{equation}
        \uqc{L}(\bm \theta) \ket{b_1 b_2\cdots b_n} = \bigotimes_{i=1}^n \frac{1}{\sqrt{2}} \left(\ket{0} + \exp\left[\pi \ii b_i + \ii \sum_{j \in \bar{\Lambda}_{l(i)}} \theta_{ij} b_j\right] \ket{1}\right).
    \end{equation}
    
    To find the matrix element, we take the inner product with $\bra{b'_1 b'_2 \cdots b'_n}$. For each qubit $i$, the $\ket{1}$ component is selected if and only if $b'_i = 1$, which mathematically corresponds to multiplying the entire phase exponent by $b'_i$. Thus, we obtain:
    \begin{align}
        &\bra{b'_1 b'_2 \cdots b'_n} \uqc{L}(\bm \theta) \ket{b_1 b_2\cdots b_n}\notag\\
        =& \frac{1}{\sqrt{2^n}} \prod_{i=1}^n \exp\left[ \pi \ii b_i b'_i + \ii \sum_{j \in \bar{\Lambda}_{l(i)}} \theta_{ij} b'_i b_j \right]\\
        =& \frac{1}{\sqrt{2^n}} \exp\left[ \ii \pi \sum_{i=1}^n b_i b'_i + \ii \sum_{l=1}^{L-1} \sum_{i \in \Lambda_l} \sum_{j \in \Lambdacomp{l}} \theta_{ij} b'_i b_j \right].
    \end{align}
    
    To simplify the notation for the matrix elements while respecting the sequential asymmetry of the circuit (where each qubit pair interacts only once), we define a directed effective phase matrix $\tilde{\Theta}$ whose elements are given by:
    \begin{equation}
        \tilde{\theta}_{ij} = \begin{cases}
            \theta_{ij}, & \text{if } i \in \Lambda_{l(i)}, j \in \bar{\Lambda}_{l(i)} \text{ (i.e., } l(i) < l(j)) \\
            \pi, & \text{if } i = j \\
            0, & \text{otherwise (i.e., } l(i) \ge l(j)).
        \end{cases}
    \end{equation}
    Using this effective matrix, the exponent can be compactly written as a full summation over all $i, j$:
    \begin{equation}
        \bra{b'_1 b'_2 \cdots b'_n} \uqc{L}(\bm \theta) \ket{b_1 b_2\cdots b_n} = \frac{1}{\sqrt{2^n}}\exp\left[ \ii \sum_{i, j} \tilde{\theta}_{ij} b'_i b_j \right].
    \end{equation}
    
    The shift invariance requires that the probability distribution over the subgroup $V$ is identical to its cosets. First, we compute the squared probability amplitude over $V$:
    \begin{equation}
        \left|\bra{x}\uqc{L}(\bm \theta)\ket{V}\right|^2
        = \frac{1}{2^n |V|}\left|\sum_{y\in V}  \exp\left[ \ii \sum_{i, j} \tilde{\theta}_{ij} x_i y_j \right] \right|^2.
    \end{equation}
    
    Since $V \leq \mathbb{Z}_{2^n}$, $V = \{q 2^s \mid q \in \mathbb{Z} \}$ for certain $s \leq n$. Assuming a standard binary encoding where the $s$ least significant bits correspond to the trailing qubits, $y\in V$ can be denoted as $y = y_1y_2\cdots y_{n-s} 0 \cdots 0$, with $y_i \in \{0,1\}$. Similarly, the coset $c + V = \{q 2^s  + c\mid 0 \le c < 2^s, q \in \mathbb{Z} \}$, and its elements $y\in c+V$ can be denoted by matching the trailing bits to the binary representation of $c$, yielding $y = y_1y_2\cdots y_{n-s} c_{n-s+1} \cdots c_n$, with $y_i, c_k \in \{0,1\}$.
    
    Therefore, we can expand the squared amplitude for the subgroup $V$:
    \begin{align}
        &\left|\bra{x}\uqc{L}(\bm \theta)\ket{V}\right|^2 \notag \\
        =& \frac{1}{2^n |V|}\left|\sum_{y\in V}  \exp\left[ \ii \sum_{i=1}^n \sum_{j=1}^n \tilde{\theta}_{ij} x_i y_j \right] \right|^2 \\
        =& \frac{1}{2^n 2^{n-s}}\left| \sum_{y_1, \dots, y_{n-s} \in \{0,1\}} \exp\left[ \ii \sum_{i=1}^n \sum_{j=1}^{n-s} \tilde{\theta}_{ij} x_i y_j \right] \right|^2.
    \end{align}
    
    Evaluating the same expression for the coset $c+V$, we factor out the terms involving the constant bits $c_k$:
    \begin{align}
        &\left|\bra{x}\uqc{L}(\bm \theta)\ket{c+V}\right|^2 \notag \\
        =& \frac{1}{2^n |c+V|}\left|\sum_{y\in c+V}  \exp\left[ \ii \sum_{i=1}^n \sum_{j=1}^n \tilde{\theta}_{ij} x_i y_j \right] \right|^2 \\
        =& \frac{1}{2^n 2^{n-s}}\left| \exp\left[ \ii \sum_{i=1}^n \sum_{k=n-s+1}^n \tilde{\theta}_{ik} x_i c_k \right] \sum_{y_1, \dots, y_{n-s} \in \{0,1\}} \exp\left[ \ii \sum_{i=1}^n \sum_{j=1}^{n-s} \tilde{\theta}_{ij} x_i y_j \right] \right|^2 \\
        =& \frac{1}{2^n 2^{n-s}}\left| \exp\left[ \ii \sum_{i=1}^n \sum_{k=n-s+1}^n \tilde{\theta}_{ik} x_i c_k \right] \right|^2 \left| \sum_{y_1, \dots, y_{n-s} \in \{0,1\}} \exp\left[ \ii \sum_{i=1}^n \sum_{j=1}^{n-s} \tilde{\theta}_{ij} x_i y_j \right] \right|^2.
    \end{align}
    
    Because the modulus of the extracted pure phase factor strictly equals $1$, it completely cancels out, yielding:
    \begin{align}
        \left|\bra{x}\uqc{L}(\bm \theta)\ket{c+V}\right|^2
        =& \frac{1}{2^n 2^{n-s}}\left| \sum_{y_1, \dots, y_{n-s} \in \{0,1\}} \exp\left[ \ii \sum_{i=1}^n \sum_{j=1}^{n-s} \tilde{\theta}_{ij} x_i y_j \right] \right|^2 \\
        =&\left|\bra{x}\uqc{L}(\bm \theta)\ket{V}\right|^2.
    \end{align}
    This completes the proof.
\end{proof}

\section{Sampling calculation and recovery guarantee for the \(\myqc-0\) case}

\label{app:transversal_hadamard_sampling}
We now give the calculation underlying the \myqc-0 discussion in the main text.
Consider $G\simeq\mathbb{Z}_{2^n}$ represented by $n$ computational-basis qubits. 
Because $\mathbb{Z}_{2^n}$ is a cyclic group of order $2^n$, any of its subgroups must be generated by a divisor of $2^n$. 
Thus, the hidden subgroup is necessarily of the form $V = \langle 2^p \rangle$ for some integer $0 \le p \le n$. 
Since every element $h \in V$ is a multiple of $2^p$, its $n$-bit binary representation naturally ends with $p$ zeros. 
Under this standard convention, an element $h$ takes the form
\begin{equation}
    h=h_1h_2\cdots h_{n-p}0\cdots 0.
\end{equation}
The subgroup generated by $2^p$ is
\begin{equation}
    V=\{q2^p:q=0,\ldots,2^{n-p}-1\},
\end{equation}
so $|V|=2^{n-p}$.
Since adding an element of $V$ can vary the first $n-p$ bits, each coset has a representative of the form
\begin{equation}
    c=0\cdots 0c_{n-p+1}\cdots c_n.
\end{equation}
Thus a normalized coset state can be written as
\begin{equation}
    \ket{c+V}
    =
    \frac{1}{\sqrt{2^{n-p}}}
    \sum_{h_1,\ldots,h_{n-p}\in\{0,1\}}
    \ket{h_1\cdots h_{n-p}c_{n-p+1}\cdots c_n}.
    \label{eq:coset_state_z2n}
\end{equation}

For an $n$-bit string $x$, applying Hadamard gates to all $n$ qubits gives
\begin{equation}
    H^{\otimes n}\ket{x}
    =
    \frac{1}{\sqrt{2^n}}
    \sum_{k\in\{0,1\}^n}(-1)^{k\cdot x}\ket{k},
\end{equation}
where $k \cdot x = \sum_{i=1}^n k_i x_i$ denotes the standard bitwise inner product of the binary strings $k$ and $x$. Applying this to Eq.~\eqref{eq:coset_state_z2n} gives
\begin{align}
    H^{\otimes n}\ket{c+V}
    &=
    \frac{1}{\sqrt{2^{n-p}}}\frac{1}{\sqrt{2^n}}
    \sum_{k\in\{0,1\}^n}
    \sum_{h_1,\ldots,h_{n-p}\in\{0,1\}}
    (-1)^{\sum_{i=1}^{n-p}k_ih_i+\sum_{i=n-p+1}^{n}k_ic_i}
    \ket{k} \nonumber\\
    &=
    \frac{1}{\sqrt{2^{n-p}}}\frac{1}{\sqrt{2^n}}
    \sum_{k\in\{0,1\}^n}
    (-1)^{\sum_{i=n-p+1}^{n}k_ic_i}
    \prod_{i=1}^{n-p}\left(\sum_{h_i\in\{0,1\}}(-1)^{k_ih_i}\right)
    \ket{k}.
\end{align}
The inner product factor is zero unless $k_1=\cdots=k_{n-p}=0$, in which case it equals $2^{n-p}$.
Therefore
\begin{equation}
    H^{\otimes n}\ket{c+V}
    =
    \frac{1}{\sqrt{2^p}}
    \sum_{k_{n-p+1},\ldots,k_n\in\{0,1\}}
    (-1)^{\sum_{i=n-p+1}^{n}k_ic_i}
    \ket{0}^{\otimes(n-p)}\ket{k_{n-p+1}\cdots k_n}.
    \label{eq:hadamard_coset_z2n}
\end{equation}
The coset representative $c$ is eliminated in Eq.~\eqref{eq:hadamard_coset_z2n}.
Consequently, the measurement distribution is independent of $c$. The first $n-p$ bits are always zero, while the remaining $p$ bits are uniformly random. 
\begin{algorithm}[t]
\caption{\(\myqc-0\) Sampling for $V\le \mathbb{Z}_{2^n}$}
\label{alg:transversal_hadamard_sampling}
\KwIn{Oracle access to an HSP instance over $\mathbb{Z}_{2^n}$, sample count $m$}
\KwOut{Hidden subgroup $V=\langle 2^p\rangle$}

Initialize an empty list of samples $\mathcal{B}$\;
\For{$i=1$ \KwTo $m$}{
    Prepare the uniform superposition $\frac{1}{\sqrt{2^n}}\sum_{g\in \mathbb{Z}_{2^n}}\ket{g}\ket{0}$\;
    Query the oracle (measuring the second register to obtain a coset state $\ket{c_i+V}$ is optional)\;
    Apply $H^{\otimes n}$ to the first register\;
    Measure the first register and record the bit string $b^{(i)}\in\{0,1\}^n$\;
    Append $b^{(i)}$ to $\mathcal{B}$\;
}
Find the minimum index $j^*$ (where $j=1$ is the most significant bit) such that $b^{(i)}_{j^*}=1$ for some sample in $\mathcal{B}$\;
\eIf{such a $j^*$ exists}{
    Let $p = n - j^* + 1$\;
}{
    Let $p = 0$\;
}
\Return{$V=\langle 2^p\rangle$}\;
\end{algorithm}

This immediately gives a simple reconstruction procedure for the idealized $\mathbb{Z}_{2^n}$ HSP.
Each application of the oracle-and-measurement routine prepares a random coset state, and Eq.~\eqref{eq:hadamard_coset_z2n} shows that subsequently applying $H^{\otimes n}$ converts that coset state into a sample whose first $n-p$ bits are deterministically zero while the remaining $p$ bits are unbiased.
The resulting sampling-and-recovery routine is summarized in Algorithm~\ref{alg:transversal_hadamard_sampling}.

We note that the physical measurement of the second register in Algorithm~\ref{alg:transversal_hadamard_sampling} is not strictly necessary. By the principle of deferred measurement, tracing out the second register leaves the first register in a mixed state of all possible cosets. Since the measurement distribution is strictly independent of the coset representative $c$, this yields identical measurement statistics. This device-friendly feature avoids the need for mid-circuit measurements, which is highly advantageous for implementations on near-term quantum hardware.

Algorithm~\ref{alg:transversal_hadamard_sampling} is correct because, under the bit-ordering convention above, the free positions occur exactly in the final $p$ bits.
Equivalently, the measurement data have the form
\begin{equation}
    0\cdots 0 * \cdots *,
\end{equation}
where each star denotes an unbiased bit.
Therefore, the classical post-processing only needs to identify the leftmost bit position that ever evaluates to $1$ across all $m$ samples. This single position $j^*$ uniquely determines $p$, and hence determines the subgroup $V=\langle 2^p\rangle$.

The sample complexity is logarithmic in the number of qubits.
Under this leftmost-$1$ decoding strategy, the algorithm only fails if the most significant free bit (at position $n-p+1$) is measured as $0$ in all $m$ independent samples.
Because this specific bit is uniformly random, the probability of it being zero $m$ times in a row is exactly $2^{-m}$.
Thus $m\ge \lceil \log_2(1/\epsilon)\rceil$ samples are sufficient for Algorithm~\ref{alg:transversal_hadamard_sampling} to perfectly identify the hidden subgroup $V$ with probability at least $1-\epsilon$. This tight bound elegantly removes the union-bound dependence on $n$, yielding an optimized sample complexity.

\section{Fisher information for subgroup-generator estimation}
\label{app:qft_fisher_information}
To benchmark the proposed circuit family, we treat the hidden subgroup generator $r$ as the parameter to be inferred from the measurement distribution produced by the standard QFT. Consider the truncated periodic state
\begin{equation}
    R:=\left\lfloor\frac{2^n-1}{r}\right\rfloor+1,
\end{equation}
where $R$ is the number of supported multiples of $r$, and
\begin{equation}
    \ket{\psi_r}
    :=
    \frac{1}{\sqrt{R}}
    \sum_{q=0}^{R-1}\ket{qr}.
\end{equation}
Thus $q=0,\ldots,R-1$ labels precisely the integers $qr$ contained in the
$2^n$-dimensional register.
For the $2^n$-point QFT,
\begin{equation}
    \ket{j}
    \mapsto
    \frac{1}{\sqrt{2^n}}
    \sum_{k=0}^{2^n-1}
    e^{2\pi \ii jk/2^n}\ket{k},
\end{equation}
and therefore
\begin{equation}
    \ket{\psi_r}
    \mapsto
    \frac{1}{\sqrt{R2^n}}
    \sum_{k=0}^{2^n-1}
    \left(
        \sum_{q=0}^{R-1}
        e^{2\pi \ii q r k/2^n}
    \right)\ket{k}.
\end{equation}
The corresponding measurement probability is
\begin{equation}
    \Pr(k|r)
    =
    \frac{1}{R2^n}
    \left|
        \sum_{q=0}^{R-1}
        e^{2\pi \ii q r k/2^n}
    \right|^2
    =
    \frac{1}{R2^n}
    \left(
        \frac{\sin(\pi R r k/2^n)}{\sin(\pi r k/2^n)}
    \right)^2.
\end{equation}

In the continuum approximation, the Fisher-information calculation reduces to a Dirichlet-kernel moment estimate. The result needed for the main-text comparison is summarized in the following lemma.

\begin{lemma}[QFT Fisher-information]
\label{thm:qft_fisher_information}
    In the regime $1\ll r\ll 2^n$, the QFT Fisher information is approximated by 
    \begin{equation}
        \fii_{\mathrm{QFT}}(r,n)
        \approx
        \frac{4\pi^2}{9}\left(\frac{2^{2n}}{r^2}-1\right).
    \end{equation}
\end{lemma}

\begin{proof}
    In the regime $r \ll 2^n$, introduce the continuum measurement coordinate
    \begin{equation}
        x:=\frac{\pi k}{2^n}\in[0,\pi].
    \end{equation}
    The corresponding continuum approximation to the QFT probability density is
    \begin{equation}
        f(x,r):=
        \frac{1}{\pi R}
        \left(
            \frac{\sin(Rrx)}{\sin(rx)}
        \right)^2.
    \end{equation}
    The QFT Fisher information with respect to the generator $r$ is
    \begin{equation}
        \fii_{\mathrm{QFT}}(r,n)
        =
        \int_{0}^{\pi}
        f(x,r)
        \left[
            \pdv{}{r}\ln f(x, r)
        \right]^2
        dx,
    \end{equation}
    with score function
    \begin{equation}
        s_r(x)
        :=
        \pdv{}{r}\ln f(x,r)
        =
        2Rx\cot(Rrx)-2x\cot(rx),
        \qquad
        s_r(x)=\text{score for the generator }r.
    \end{equation}
    Define
    \begin{equation}
        D_R(u) := \frac{\sin(Ru)}{\sin u},
        \qquad
        D_R(u)=\text{order-}R\text{ Dirichlet kernel}.
    \end{equation}
    Then $f(x, r)=D_R(rx)^2/(\pi R)$, and the score function implies
    \begin{equation}
        f(x, r)\left[\pdv{}{r}\ln f(x, r)\right]^2
        =
        \frac{4x^2}{\pi R}\left[D_R'(rx)\right]^2.
    \end{equation}
    Thus the Fisher-information integral is reduced to the square of a finite trigonometric polynomial. Once that polynomial is expanded into cosine modes, the remaining calculation involves only the moments of $x^2$ against those modes.

    We begin by expressing the Dirichlet kernel $D_R(u)$ in its exponential sum form. By shifting the summation index to center the phases, it can be written as a symmetric finite series:
    \begin{equation}
        D_R(u)=\frac{\sin(Ru)}{\sin u}
        =
        \sum_{m=0}^{R-1}e^{\ii(R-1-2m)u},
    \end{equation}
    Differentiating this series term-by-term with respect to $u$ yields:
    \begin{equation}
        D_R'(u)
        =
        \ii\sum_{m=0}^{R-1}(R-1-2m)e^{\ii(R-1-2m)u}.
    \end{equation}
    To compute the squared magnitude $[D_R'(u)]^2$, we multiply this sum by itself. By collecting terms that produce the same relative frequency difference, we can perfectly recast the product as an even cosine series:
    \begin{equation}
        [D_R'(u)]^2
        =
        A_0 + 2\sum_{j=1}^{R-1}A_j\cos(2ju),
    \end{equation}
    where $A_j$ denotes the cosine-series coefficient at frequency $2j$. These
    coefficients are obtained by evaluating the convolution sum over amplitudes
    with frequency offset $j$. To evaluate this sum explicitly, we expand the
    summand as a quadratic in $m$ and apply the standard summation formulas for
    $\sum m$ and $\sum m^2$:
    \begin{align}
        A_j
        &= -\sum_{m=0}^{R-1-j} (R-1-2m)(2j-R+1+2m) \nonumber\\
        &= \sum_{m=0}^{R-1-j} \left[ 4m^2 - 4(R-1-j)m + (R-1-j)^2 - j^2 \right] \nonumber\\
        &= \frac{R-j}{3} \left[ (R-j-1)^2 + 2(R-j-1) - 3j^2 \right] \nonumber\\
        &= \frac{(R-j)(R^2-2Rj-2j^2-1)}{3}, \qquad j=0,\ldots,R-1.
    \end{align}
    In particular, setting $j=0$ isolates the non-oscillatory constant term:
    \begin{equation}
        A_0=\frac{R(R^2-1)}{3}.
    \end{equation}
    Recall from our earlier derivation of the score function that the Fisher information integral takes the form:
    \begin{equation}
        \fii_{\mathrm{QFT}}(r,n)
        =
        \frac{4}{\pi R}\int_0^\pi x^2[D_R'(rx)]^2\,dx,
    \end{equation}
    Substituting our exact cosine expansion of $[D_R'(rx)]^2$ into this integral distributes the integration across the individual frequency modes:
    \begin{equation}
        \fii_{\mathrm{QFT}}(r,n)
        =
        \frac{4}{\pi R}
        \left[
            A_0\int_0^\pi x^2\,dx
            +
            2\sum_{j=1}^{R-1}A_j\int_0^\pi x^2\cos(2jrx)\,dx
        \right].
    \end{equation}
    The problem is now reduced to evaluating two elementary definite integrals
    over $[0,\pi]$. Define the non-oscillatory and oscillatory moments by
    \begin{align}
        M_0&:=\int_0^\pi x^2\,dx=\frac{\pi^3}{3},\\
        M_j(r)&:=\int_0^\pi x^2\cos(2jrx)\,dx \nonumber\\
        &=
        \frac{
            4\pi jr\cos(2\pi jr)
            +(4\pi^2j^2r^2-2)\sin(2\pi jr)
        }{(2jr)^3},
        \qquad j\geq 1.
    \end{align}
    Here $M_0$ is the non-oscillatory second moment, and $M_j(r)$ is the
    oscillatory second moment at mode $j$.
    Inserting these evaluated moments, along with the explicit expressions for $A_0$ and $A_j$, back into the expanded Fisher information equation gives:
    \begin{align}
        \fii_{\mathrm{QFT}}(r,n)
        &=
        \frac{4}{\pi R}
        \left[
            \frac{R(R^2-1)}{3}\cdot\frac{\pi^3}{3}
            +
            \frac{2}{3}
            \sum_{j=1}^{R-1}
            (R-j)(R^2-2Rj-2j^2-1)\,M_j(r)
        \right] \nonumber\\
        &=
        \frac{4\pi^2}{9}(R^2-1)
        +
        \frac{8}{3\pi R}
        \sum_{j=1}^{R-1}
        (R-j)(R^2-2Rj-2j^2-1)\,M_j(r),
    \end{align}
    which is the stated exact formula. 
    
    Finally, we analyze the asymptotic behavior in the standard regime of interest, $1 \ll r \ll R$. In this limit, the hidden parameter is large enough to suppress oscillations, yet small enough that its multiples are heavily supported in the register. The explicit form of the oscillatory moments shows that $M_j(r)=\mathcal{O}(r^{-2})$ for any fixed mode $j$. Since the leading constant term scales as $\mathcal{O}(R^2)$ and the summation over the modes scales proportionally to $\mathcal{O}(R^2 r^{-2})$, the relative magnitude of the oscillatory corrections is heavily suppressed by a factor of $1/r^2$. Because these high-frequency corrections are strictly subleading in this large-$r$ regime, the behavior is overwhelmingly governed by the constant contribution. Keeping only this non-oscillatory term and substituting $R \approx 2^n/r$ therefore gives:
    \begin{equation}
        \fii_{\mathrm{QFT}}(r,n)
        \approx
        \frac{4\pi^2}{9}\left(\frac{2^{2n}}{r^2}-1\right).
        \label{eq:qft_dfi_asymptotic}
    \end{equation}
    This proves the lemma.
\end{proof}

\color{revisioncolor}

\begin{lemma}[HP-1 period-law overlap]
\label{lem:hp1_period_law_overlap}
For the fixed-phase \myqc-1 circuit, let \(\mathrm{Pr}(\cdot\mid s)\) and
\(\mathrm{Pr}(\cdot\mid t)\) denote the coherent output probability laws associated with two candidate
periods \(s,t\).  Let \(\nu_2(m)\) denote the largest integer \(a\) such that
\(2^a\) divides \(m\), and set \(a=\nu_2(\gcd(s,t))\).  Then, for a constant
\(\kappa_a\) depending on \(a\), their overlap obeys the following statistical
law:
\begin{equation}
    \sum_x \mathrm{Pr}(x\mid s)\mathrm{Pr}(x\mid t)
    \approx
    \frac{1}{2^n}
    +
    \kappa_a\,\frac{2^a}{st},
    \qquad
    a=\nu_2(\gcd(s,t)).
    \label{eq:hp1_period_law_overlap_placeholder}
\end{equation}
\end{lemma}

\begin{proof}
Set \(N:=2^n\).
The goal is to evaluate the period-law overlap by expanding both coherent
probability laws and then carrying out the sum over output strings \(x\).
This separates the overlap into a diagonal contribution, two single-period
off-diagonal contributions, and one genuinely mixed off-diagonal contribution.
For the fixed-phase \myqc-1 circuit, let
\(C=C_{\mathrm{HP1}}\in\mathbb R^{n\times n}\) denote the HP-1 phase matrix
normalized by
\begin{equation}
    \bra{x}\uqc{1}\ket{y}
    =
    \frac{1}{\sqrt N}
    \exp\!\left(\pi\ii\,x^T C y\right),
    \qquad x,y\in\{0,1\}^n .
    \label{eq:hp1_overlap_phase_matrix_definition}
\end{equation}
Equivalently, \(C\) is the directed fixed-phase HP-1 matrix with nonzero
entries
\begin{equation}
    C_{ij}
    =
    \begin{cases}
        \dfrac{1}{2^{|i-j|}},
        & i=j\ \text{or}\ (i\in\Lambda_1,\ j\in\Lambda_2),\\[1mm]
        0, & \text{otherwise},
    \end{cases}
    \label{eq:hp1_overlap_phase_matrix_entries}
\end{equation}
For a candidate period \(s\), define
\begin{equation}
    R_s
    :=
    \left|\{q\in\mathbb Z_{\ge0}:0\le qs<N\}\right|
    =
    \left\lfloor\frac{N-1}{s}\right\rfloor+1 .
    \label{eq:hp1_overlap_period_count}
\end{equation}
Let \(\mathbf b_m\in\{0,1\}^n\) denote the \(n\)-bit binary representation of
the integer \(m\).  Using Eq.~\eqref{eq:hp1_overlap_phase_matrix_definition},
the probability law \(\mathrm{Pr}(\cdot\mid s)\) is
\begin{align}
    \mathrm{Pr}(x\mid s)
    &=
    \left|
        \frac{1}{\sqrt{R_s}}
        \sum_{q=0}^{R_s-1}
        \bra{x}\uqc{1}\ket{\mathbf b_{qs}}
    \right|^2 \nonumber\\
    &=
    \frac{1}{NR_s}
    \left|
        \sum_{q=0}^{R_s-1}
        \exp\!\left(\pi\ii\,x^T C\mathbf b_{qs}\right)
    \right|^2,
    \qquad x\in\{0,1\}^n .
    \label{eq:hp1_overlap_probability_law}
\end{align}
Substituting Eq.~\eqref{eq:hp1_overlap_probability_law} and the corresponding
expression for \(\mathrm{Pr}(\cdot\mid t)\) into the overlap on the left-hand side of
Eq.~\eqref{eq:hp1_period_law_overlap_placeholder} gives
	\begin{align}
	    \sum_x \mathrm{Pr}(x\mid s)\mathrm{Pr}(x\mid t)
	    &=
	    \frac{1}{N^2R_sR_t}
    \sum_{x\in\{0,1\}^n}
    \left|
        \sum_{q=0}^{R_s-1}
        \exp\!\left(\pi\ii\,x^T C\mathbf b_{qs}\right)
    \right|^2 \nonumber\\
    &\quad\times
    \left|
        \sum_{\ell=0}^{R_t-1}
        \exp\!\left(\pi\ii\,x^T C\mathbf b_{\ell t}\right)
    \right|^2 .
    \label{eq:hp1_overlap_substituted_probability_laws}
	\end{align}
We next perform the \(x\)-sum exactly.  Since \(x\) ranges over all binary
strings, the sum factorizes bit by bit.  The following kernel records the real
contribution of each conjugate pair that appears after expanding the two
squared moduli.
For \(v\in\mathbb Z^n\), the associated conjugate-pair identity defines the
real cosine kernel \(\Omega(v)\) by
\begin{equation}
    \begin{aligned}
    \prod_{a=1}^n
    \left[1+e^{\pi\ii(Cv)_a}\right]
    +
    \prod_{a=1}^n
    \left[1+e^{-\pi\ii(Cv)_a}\right]
    &=
    2N
    \cos\!\left[
        \frac{\pi}{2}
        \sum_{a=1}^n
        \bigl(Cv\bigr)_a
    \right]
    \prod_{a=1}^n
    \cos\!\left[
        \frac{\pi}{2}
        \bigl(Cv\bigr)_a
    \right]
    =:
    2N\Omega(v).
    \end{aligned}
    \label{eq:hp1_overlap_real_cosine_kernel}
\end{equation}
	For fixed \(q,q',\ell,\ell'\), set
	\[
	    \Delta_{q,q',\ell,\ell'}
    :=
    \mathbf b_{qs}-\mathbf b_{q's}
    +\mathbf b_{\ell t}-\mathbf b_{\ell't}.
\]
Expanding the two squared moduli in
Eq.~\eqref{eq:hp1_overlap_substituted_probability_laws}, exchanging the order
of summation, and substituting the \(x\)-sum kernel explicitly gives
\begin{align}
    \sum_x \mathrm{Pr}(x\mid s)\mathrm{Pr}(x\mid t)
    &=
    \frac{1}{N^2R_sR_t}
    \sum_{q,q'=0}^{R_s-1}
    \sum_{\ell,\ell'=0}^{R_t-1}
    \prod_{a=1}^n
    \left[
        1+\exp\!\left(
            \pi\ii\,\bigl(C\Delta_{q,q',\ell,\ell'}\bigr)_a
        \right)
    \right] \nonumber\\
    &=
    \Bigg[
        \frac{1}{N}
        +\frac{2}{NR_s}
        \sum_{0\le q<q'<R_s}
        \Omega\!\left(\Delta_{q,q',0,0}\right)
        +\frac{2}{NR_t}
        \sum_{0\le \ell<\ell'<R_t}
        \Omega\!\left(\Delta_{0,0,\ell,\ell'}\right)
        \nonumber\\
    &\qquad\qquad
        +\frac{2}{NR_sR_t}
        \sum_{0\le q<q'<R_s}
        \sum_{0\le \ell<\ell'<R_t}
        \left[
            \Omega\!\left(\Delta_{q,q',\ell,\ell'}\right)
            +
            \Omega\!\left(\Delta_{q,q',\ell',\ell}\right)
        \right]
    \Bigg].
	    \label{eq:hp1_overlap_x_sum_first}
	\end{align}
\begin{figure}[!t]
    \centering
    \includegraphics[width=0.8\columnwidth]{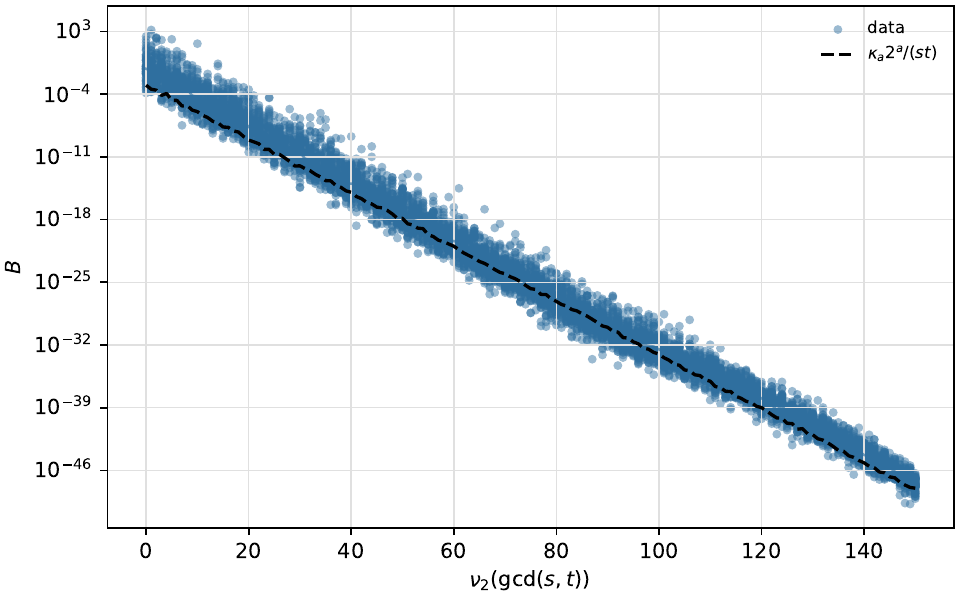}
    \caption{\label{fig:hp1_mixed_B_statistical_law}
    \textbf{Linear regression between\(B_{s,t}^{(n)}\) and $\nu_2$.  }
    Conditioned Monte Carlo estimate of \(B_{s,t}^{(n)}\) for fixed-phase
    \myqc-1 at \(n=200\).  For each \(a=0,\ldots,150\), \(100\) dyadic
    pairs are sampled.  The dashed curve shows
    Eq.~\eqref{eq:hp1_mixed_B_statistical_law}.}
\end{figure}
Here \(\Delta_{q,q',0,0}=\mathbf b_{qs}-\mathbf b_{q's}\) and
\(\Delta_{0,0,\ell,\ell'}=\mathbf b_{\ell t}-\mathbf b_{\ell't}\), since the
paired diagonal difference is zero.
The first equality is the fully substituted form of the \(x\)-sum: no additional
kernel notation is being used.  The passage from the first line of
Eq.~\eqref{eq:hp1_overlap_x_sum_first} to the second line comes from grouping
ordered indices into the smallest real orbits.  The conjugate-pair identity in
Eq.~\eqref{eq:hp1_overlap_real_cosine_kernel} is applied to each conjugate pair.
If \(q=q'\) and \(\ell=\ell'\), then \(\Delta=0\), the product equals \(N\),
and the \(R_sR_t\) such terms give the first normalized contribution \(1/N\)
inside the brackets.
If \(q\ne q'\) while \(\ell=\ell'\), then for each unordered pair \(q<q'\)
and each of the \(R_t\) diagonal values of \(\ell\), the two ordered terms
\((q,q')\) and \((q',q)\) have opposite vectors
\(\pm(\mathbf b_{qs}-\mathbf b_{q's})\), giving the normalized coefficient
\(2/(NR_s)\).
The case \(q=q'\), \(\ell\ne\ell'\) is identical and gives the coefficient
\(2/(NR_t)\).  Finally, when both ordered pairs are off diagonal, the four
terms generated by
\((q,q')\leftrightarrow(q',q)\) and
\((\ell,\ell')\leftrightarrow(\ell',\ell)\) split into the two conjugate pairs
\(\pm\Delta_{q,q',\ell,\ell'}\) and
\(\pm\Delta_{q,q',\ell',\ell}\), giving the final double sum with normalized
coefficient \(2/(NR_sR_t)\).  Thus the remaining sums run over unordered indices
\(q<q'\) and \(\ell<\ell'\).  Individual \(\Omega\)-terms need not be
nonnegative; the nonnegativity of the whole expression follows from its equality
to the probability overlap \(\sum_x\mathrm{Pr}(x\mid s)\mathrm{Pr}(x\mid t)\).

The two single-period off-diagonal sums in
Eq.~\eqref{eq:hp1_overlap_x_sum_first} are in fact killed by a zero factor in
the cosine product.  For \(q\ne q'\), set
\begin{equation}
    u
    =
    \mathbf b_{qs}-\mathbf b_{q's}.
    \label{eq:hp1_single_s_difference_vector}
\end{equation}
Since \(u\ne0\) and each component of \(u\) belongs to \(\{-1,0,1\}\), some
component \(u_a\) is odd.  If a changed component lies in \(\Lambda_2\), then
row \(a\) of \(C\) is diagonal and gives the zero factor
\(\cos(\pi u_a/2)=0\).  If no component in \(\Lambda_2\) changes, then some
component in \(\Lambda_1\) changes; all \(\Lambda_2\)-tails vanish in that row,
so again \((Cu)_a=u_a=\pm1\) and the same zero factor appears.  The same
argument with \(q=q'\) gives the corresponding single-\(t\) cancellation:
\begin{equation}
    \begin{aligned}
    \Omega\!\left(\Delta_{q,q',0,0}\right)&=0,
    && q\ne q',\\
    \Omega\!\left(\Delta_{0,0,\ell,\ell'}\right)&=0,
    && \ell\ne\ell'.
    \end{aligned}
    \label{eq:hp1_single_omega_vanishes}
\end{equation}
The mixed term is different because two such bit-difference vectors are added.
Writing
\begin{equation}
    \Delta_{q,q',\ell,\ell'}
    =
    u+w,
    \qquad
    w=\mathbf b_{\ell t}-\mathbf b_{\ell't},
    \label{eq:hp1_mixed_difference_sum}
\end{equation}
one has
\begin{equation}
    u_a,w_a\in\{-1,0,1\},
    \qquad
    \bigl(\Delta_{q,q',\ell,\ell'}\bigr)_a
    \in\{-2,-1,0,1,2\}.
    \label{eq:hp1_mixed_difference_components}
\end{equation}
For the mixed off-diagonal contribution, we define the rescaled quantity
\begin{equation}
    B_{s,t}^{(n)}
    :=
    \frac{1}{2^{n-1}R_sR_t}
    \sum_{0\le q<q'<R_s}
    \sum_{0\le \ell<\ell'<R_t}
    \left[
        \Omega\!\left(\Delta_{q,q',\ell,\ell'}\right)
        +
        \Omega\!\left(\Delta_{q,q',\ell',\ell}\right)
    \right].
    \label{eq:hp1_mixed_B_definition}
\end{equation}
This is the only term in Eq.~\eqref{eq:hp1_overlap_x_sum_first} that is not
forced to vanish by the single-period cancellation.  The normalization by
\(2^{n-1}R_sR_t\) removes the explicit system-size factor from the mixed double
sum, so that different \(n\) and different dyadic scales can be compared on
the same plot.
Numerically, we observe a simple statistical law for this mixed contribution.
For dyadic periods \(s=2^a u\) and \(t=2^a v\), with \(u,v\) odd and
\(a=\nu_2(\gcd(s,t))\), the conditioned Monte Carlo estimate in
Fig.~\ref{fig:hp1_mixed_B_statistical_law} is well described by
\begin{equation}
    B_{s,t}^{(n)}
    \approx
    \kappa_a\,\frac{2^a}{st}.
    \label{eq:hp1_mixed_B_statistical_law}
\end{equation}
Thus Fig.~\ref{fig:hp1_mixed_B_statistical_law} supplies the empirical closure
for the remaining mixed double sum: after the above normalization, its leading
dependence is \(2^a/(st)\), up to the coefficient \(\kappa_a\).
At the level of bit differences, an odd change in one period difference can be
cancelled by the other difference, or two equal changes can combine into an even
component, for example \(1+(-1)=0\), \(1+1=2\), and \((-1)+(-1)=-2\).
These cancellations can remove the zero cosine factor, so the mixed
\(\Omega\)-terms are not identically zero, even though they are numerically
sparse for random off-diagonal indices.

Finally, substituting the single-period cancellation
Eq.~\eqref{eq:hp1_single_omega_vanishes} and the statistical law
Eq.~\eqref{eq:hp1_mixed_B_statistical_law} into
Eq.~\eqref{eq:hp1_overlap_x_sum_first}, together with the definition of
\(B_{s,t}^{(n)}\) in Eq.~\eqref{eq:hp1_mixed_B_definition}, gives
\begin{equation}
    \sum_x \mathrm{Pr}(x\mid s)\mathrm{Pr}(x\mid t)
    \approx
    \frac{1}{2^n}
    +
    \kappa_a\,\frac{2^a}{st},
    \qquad
    a=\nu_2(\gcd(s,t)).
    \label{eq:hp1_overlap_statistical_law}
\end{equation}
This is the overlap formula stated in Eq.~\eqref{eq:hp1_period_law_overlap_placeholder}.
\end{proof}

\begin{theoremrestatement}{thm:hp1_first_order_fisher_transfer_maintext}{ln-linear DFI estimation}
For fixed-phase \myqc-1, the minimum DFI follows the ln-linear lower bound
under \(2<r<2^{n/4}\),
\begin{equation}
    \fii_{\min}(n)
    \gtrsim
    \exp(kn+b),
    \label{eq:hp1_restricted_window_ln_linear_bound_appendix}
\end{equation}
for constants \(k\) and \(b\).
\end{theoremrestatement}

\begin{proof}[Proof of Theorem~\ref{thm:hp1_first_order_fisher_transfer_maintext}]
Set \(N:=2^n\).  We begin directly from the DFI in Eq.~\eqref{eq:DFI}:
\begin{equation*}
    \fii(r,n)
    =
    \sum_{x\in\{0,1\}^n}
    \frac{\left[\Pr(x\mid r+1)-\Pr(x\mid r)\right]^2}
    {\Pr(x\mid r)}.
\end{equation*}
Write \(\Delta_r(x):=\Pr(x\mid r+1)-\Pr(x\mid r)\).  Since every DFI
summand is nonnegative, we may retain only those outputs for which the
denominator is exponentially small, \(\Pr(x\mid r)<2/N=2^{1-n}\), while the
finite difference is sufficiently large.  For any constant \(\tau>0\),
independent of \(n\) and \(r\), define
\begin{equation}
    A_\tau(n,r)
    :=
    \left\{
        x\in\{0,1\}^n:
        \Pr(x\mid r)<\frac{2}{N},
        \ \Delta_r(x)^2\ge\frac{\tau}{Nr^2}
    \right\}.
    \label{eq:hp1_active_tail_set}
\end{equation}
Equivalently, the two scaled conditions used by the rare-event simulation are
\(N\Pr(x\mid r)<2\) and
\(\bigl[N\Delta_r(x)\bigr]^2r^2\ge\tau N\).
Restricting the DFI sum to this active set and applying its two defining
inequalities gives
\begin{equation}
    \fii(r,n)
    \ge \sum_{x\in A_\tau(n,r)}
    \frac{\Delta_r(x)^2}{\Pr(x\mid r)}
    \ge \frac{\tau}{2r^2}\,\lvert A_\tau(n,r)\rvert.
    \label{eq:hp1_active_tail_dfi_certificate}
\end{equation}
For completeness, Eq.~\eqref{eq:hp1_active_tail_dfi_certificate} also covers
zero denominators.  Both \(\Pr(x\mid r)\) and \(\Pr(x\mid r+1)\) are
probabilities and are therefore nonnegative.  If \(\Pr(x\mid r)>0\), we use
the fraction as usual.  If \(\Pr(x\mid r)=0\), the term is \(0\) when
\(\Pr(x\mid r+1)=0\) and \(+\infty\) when \(\Pr(x\mid r+1)>0\):
\begin{equation}
    \frac{\Delta_r(x)^2}{\Pr(x\mid r)}
    :=
    \begin{cases}
        \dfrac{\Delta_r(x)^2}{\Pr(x\mid r)},
        & \Pr(x\mid r)>0,\\[1mm]
        0,
        & \Pr(x\mid r)=\Pr(x\mid r+1)=0,\\[1mm]
        +\infty,
        & \Pr(x\mid r)=0<\Pr(x\mid r+1).
    \end{cases}
    \label{eq:hp1_exact_pearson_zero_convention}
\end{equation}
These are all the possible cases, so every term in the active-set sum is well
defined.  If \(\Pr(x\mid r)=\Pr(x\mid r+1)=0\), then
\(\Delta_r(x)=0\), so the large-numerator condition
\(\Delta_r(x)^2\ge\tau/(Nr^2)>0\) fails and the output is not in
\(A_\tau(n,r)\).  If \(\Pr(x\mid r)=0<\Pr(x\mid r+1)\) while the output is
active, its contribution is \(+\infty\), in which case the lower bound in
Eq.~\eqref{eq:hp1_active_tail_dfi_certificate} holds trivially.  Every
remaining active output has \(\Pr(x\mid r)>0\), where the ordinary pointwise estimate
\begin{equation}
    \frac{\Delta_r(x)^2}{\Pr(x\mid r)}
    \ge
    \frac{\tau/(Nr^2)}{2/N}
    =
    \frac{\tau}{2r^2}
    \label{eq:hp1_active_tail_finite_point_bound}
\end{equation}
is finite and applies directly.  Thus a zero denominator can only remove an
output from the active set or make the exact DFI infinite; it can never weaken
the counting certificate.  No \(\epsilon\)-cutoff or denominator
regularization is used in this argument.

\begin{figure}[t]
    \centering
    \includegraphics[width=0.8\linewidth]{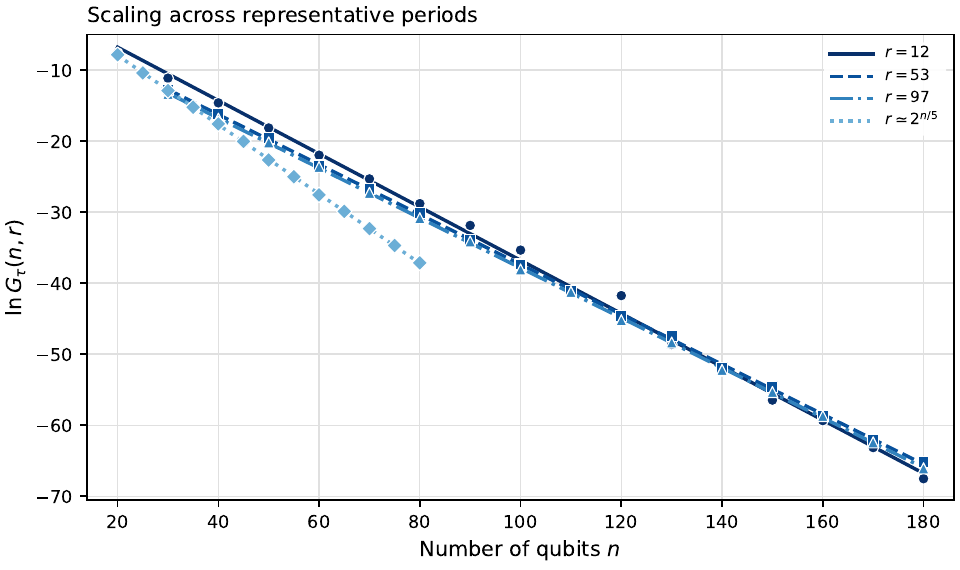}
    \caption{\label{fig:hp1_active_tail_linear_fits}
    \textbf{Log-linear scaling of the active-tail count.}
    Subset-simulation estimates of \(\ln G_\tau(n,r)\) for fixed
    \(r=12\) at \(n=20,30,\ldots,200\) (circles), and for eleven
    near-window-edge periods increasing from \(r=30\) to \(r=32000\) at
    \(n=20,24,\ldots,60\) (triangles).  Lines are
    ordinary least-squares fits to the corresponding points.  The fits
    summarize finite-window numerical evidence and are not uniform analytic
    bounds.}
\end{figure}

It remains to estimate the size of the rare event.  Introduce its normalized
count
\begin{equation}
    G_\tau(n,r)
    :=
    \frac{\lvert A_\tau(n,r)\rvert}{Nr^2}.
    \label{eq:hp1_active_tail_normalized_count}
\end{equation}
For the numerical rare-event estimate below, we now choose
\(\tau=3\times10^{-4}\).  The adaptive simulation evaluates this quantity by
subset simulation, using a log-domain roots-of-unity formula for the point
probabilities.  The
estimator was first calibrated against complete enumeration and ordinary
Monte Carlo in regimes where the event is directly resolvable.  For example,
the reference active fractions at \((n,r)=(20,12),(22,12),(30,12)\) are
\(0.06952095\), \(0.03225803\), and \(0.00213275\), respectively, while the
subset estimates lie in the corresponding ranges \(0.069\)--\(0.072\),
\(0.0312\)--\(0.0331\), and \(0.00212\)--\(0.00226\).

The resulting rare-event conclusion is the conservative statistical
lower-envelope law
\begin{equation}
    \min_{2\le r<2^{n/4}}G_\tau(n,r)
    \gtrsim
    c_A e^{-\beta n},
    \qquad
    c_A>0,
    \qquad
    \beta=0.55<\ln 2.
    \label{eq:hp1_active_tail_rare_event_law}
\end{equation}
Figure~\ref{fig:hp1_active_tail_linear_fits} visualizes three complementary
log-linear checks supporting this conservative lower-envelope choice.
Complete enumeration of the even non-dyadic part of the restricted window
gives \(\min_r\log G_\tau(n,r)=-0.499916n+1.49303\) over
\(n=14,\ldots,20\), with \(R^2=0.999554\).  For fixed \(r=12\), subset
simulation at ten-qubit intervals through \(n=200\) gives
\(\log G_\tau(n,12)=-0.375465n+0.789480\), with \(R^2=0.996496\).
To probe substantially larger periods, we additionally sample eleven points
at \(n=20,24,\ldots,60\), choosing even non-dyadic periods near the moving
window edge.  The periods increase from \(r=30\) to \(r=32000\), and each lies
between \(0.9375\) and \(0.9766\) of the upper limit \(2^{n/4}\).
These near-edge points give \(\log G_\tau=-0.526862n+2.04513\), with
\(R^2=0.999942\), and remain above \(-0.55n\) by a minimum log margin of
\(2.546\).  The near-edge sweep stops at \(n=60\) because the exact point
query scales as \(O(n[u_r+u_{r+1}])\), where \(u_s\) is the odd part of
\(s\); along \(r\simeq2^{n/4}\), this cost doubles approximately every four
qubits.  Representative
bounded-odd-part period checks at \(n=200\), together with the exact-window
data, also remain above the conservative envelope; the smallest observed log
margin over all reported checks is \(1.741\).
Equation~\eqref{eq:hp1_active_tail_rare_event_law} extrapolates this observed
lower envelope uniformly over the period window in the theorem.  It is the
explicit empirical closure of the argument, not an independent analytic
bound.

Combining Eqs.~\eqref{eq:hp1_active_tail_dfi_certificate} and
\eqref{eq:hp1_active_tail_normalized_count} yields
\begin{equation}
    \fii(r,n)
    \ge
    \frac{\tau N}{2}G_\tau(n,r).
    \label{eq:hp1_active_tail_dfi_normalized}
\end{equation}
Taking the minimum over the restricted period window and applying
Eq.~\eqref{eq:hp1_active_tail_rare_event_law}, we obtain
\begin{align}
    \fii_{\min}(n)
    &=
    \min_{2\le r<2^{n/4}}\fii(r,n) \nonumber\\
    &\gtrsim
    \frac{\tau c_A}{2}
    \exp\!\left[(\ln 2-\beta)n\right]
    =
    \exp(kn+b),
    \label{eq:hp1_restricted_window_ln_linear_bound}
\end{align}
where
\begin{equation}
    k:=\ln 2-\beta>0,
    \qquad
    b:=\ln\!\left(\frac{\tau c_A}{2}\right).
    \label{eq:hp1_active_tail_k_b}
\end{equation}
For the conservative choices \(c_A=1\) and \(\beta=0.55\), this is the
explicit simulation-supported bound
\begin{equation}
    \fii_{\min}(n)
    \gtrsim
    1.5\times10^{-4}\exp(0.143147\,n),
    \label{eq:hp1_active_tail_explicit_bound}
\end{equation}
which has the form stated in the theorem.
\end{proof}

\section{Log-linear growth fit for the numerical Fisher information}
\label{app:fi_log_linear_growth}

The fixed-$r$ QFT benchmark derived above clarifies the dominant scaling mechanism, but Fig.~\ref{fig:fi_scaling} probes a slightly different quantity. There, the vertical axis is the minimum discrete Fisher information over the period window used in the numerical scan. To maintain the asymptotic separation
\(
1 \ll r \ll 2^n
\)
as the system size increases, we use the scaling period window
\(
2 \le r \le 2^{n/2}.
\)
The relevant QFT comparator is therefore the worst-case Fisher information within this $n$-dependent window, rather than the value at a single fixed period.

The leading term in Lemma~\ref{thm:qft_fisher_information} decreases with $r$, since in the sampled regime
\(
R \approx 2^n/r.
\)
Consequently, the minimum QFT Fisher information within the scan window is attained near the largest allowed period. Taking
\(
r_{\max}(n)=2^{n/2}
\)
gives the window-matched estimate
\begin{equation}
    \fii_{\mathrm{QFT,min}}(n)
    \approx
    \frac{4\pi^2}{9}
    \left(
        \left\lfloor
        \frac{2^n}{2^{n/2}}
        \right\rfloor^2
        -1
    \right).
\end{equation}
Keeping only the dominant exponential dependence therefore gives
\begin{equation}
    \ln \fii_{\mathrm{QFT,min}}(n)
    =
    n\ln 2+\mathcal{O}(1).
\end{equation}
Thus, under the same period-window scaling used in the numerical scan, the asymptotic QFT exponent approaches
\(
k_{\mathrm{QFT},\infty}=\ln 2.
\)

This provides a natural benchmark for the numerical data in Fig.~\ref{fig:fi_scaling}. If the plotted curves exhibit an approximately exponential dependence on the qubit number over the sampled range, then their logarithms should be well described by
\begin{equation}
    \ln \fii_{\min}(n)=kn+b.
\end{equation}
We therefore perform ordinary least-squares fits to the log-transformed mean values of each curve shown in Fig.~\ref{fig:fi_scaling}. Table~\ref{tab:fi_log_regression_results} summarizes the resulting effective exponents. We do not reproduce a separate log-transformed plot because it contains the same data points as Fig.~\ref{fig:fi_scaling}.

\begin{table}[h]
    \centering
    \begin{tabular}{l c c c c}
        \hline
        Circuit & qubits & $k$ & $95\%$ CI & $R^2$ \\
        \hline
        \myqc-1 fixed & $7$--$18$ & $0.368$ & $[0.324,\,0.413]$ & $0.972$ \\
        \myqc-1 shared optimized & $7$--$18$ & $0.379$ & $[0.354,\,0.403]$ & $0.992$ \\
        \myqc-$L$ random & $7$--$18$ & $0.398$ & $[0.383,\,0.414]$ & $0.997$ \\
        \myqc-$L$ random with phase & $7$--$18$ & $0.352$ & $[0.336,\,0.367]$ & $0.996$ \\
        \hline
    \end{tabular}
    \caption{\label{tab:fi_log_regression_results} Effective exponents obtained from the log-linear model $\ln \fii_{\min}(n)=kn+b$ applied to the numerical Fisher-information curves in Fig.~\ref{fig:fi_scaling}.}
\end{table}

The fitted lines account for most of the log-scale variance across the sampled window. The coefficient of determination ($R^2$) ranges from $0.972$ to $0.997$, the log-space root-mean-square residual lies between $0.075$ and $0.217$, and every fitted slope ($k$) is significantly positive with a $p$-value $\leq 4.6\times 10^{-9}$. In that sense, a linear regression on $\ln I_{\min}$ is justified as a finite-window descriptive model for these curves. At the same time, the fitted slope $k$ should be interpreted as an effective exponent on the restricted range $n=7,\dots,18$ (where $n$ is the number of qubits), not as a claim about the ultimate asymptotic law.

The comparison with the QFT benchmark is unambiguous. Even after matching the period window used in the numerical scan, the QFT reference fit remains much steeper, with $k_{\mathrm{QFT,win}}=1.089$ and $R^2=0.999$. The empirical circuit families therefore grow substantially more slowly: their fitted exponents occupy $32\%$--$37\%$ of the window-matched QFT slope, and $25\%$--$29\%$ of the asymptotic QFT exponent $2\ln 2$. Among the plotted families, the random-layer \myqc{} circuit has the largest point estimate, near $k=0.398$, followed by the shared optimized \myqc-1 circuit, near $k=0.379$.

This comparison sharpens the interpretation of Fig.~\ref{fig:fi_scaling}. The \myqc{} family circuits retain a clear exponential-in-$n$ growth trend in the $\fii_{\text{min}}$ over the sampled window, which is why the log-linear fit works well, but they do not reproduce the QFT growth rate. Their advantage in the present setting is therefore not that they match the full QFT benchmark asymptotically, but that they retain a nontrivial and practically usable fraction of that information while remaining much shallower.

\section{Decoder Model}
\label{sec:AppendixDecoder}
\begin{figure}[!h]
    \centering
    \includegraphics[width=0.9\linewidth]{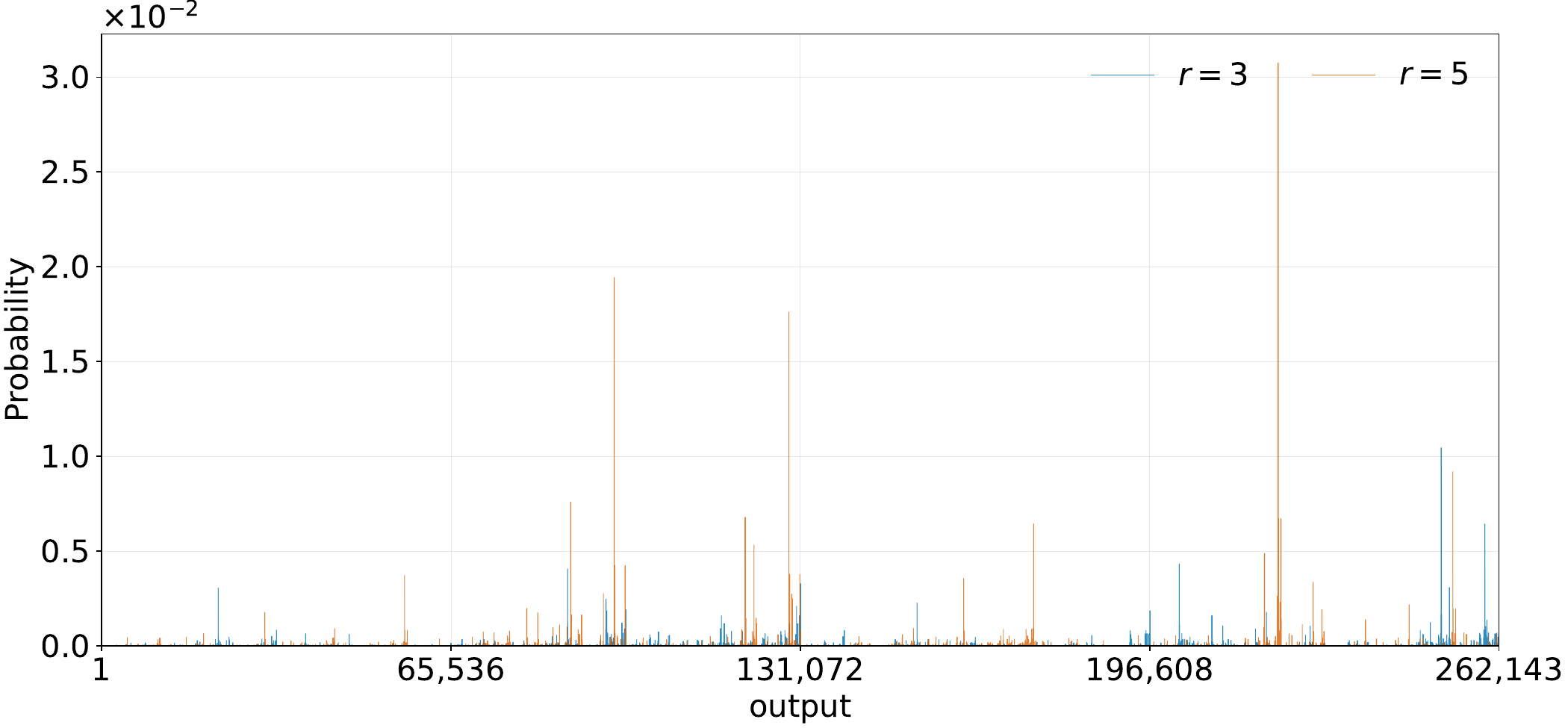}
    \caption{\textbf{Period-dependent HP output distributions.} Output
    probabilities for two hidden periods, \(r=3\) and \(r=5\).  The
    period-dependent peak structure shows that the measurement samples retain
    information about \(r\), while the irregular peak locations do not provide an
    obvious closed-form decoding rule.}
    \label{fig:decoder_output_distributions}
\end{figure}
The measurement distribution produced by an HP circuit depends on the hidden period.  For example, Fig.~\ref{fig:decoder_output_distributions} shows the output distributions for two different periods, \(r=3\) and \(r=5\).  The locations and weights of the prominent peaks change with \(r\), demonstrating that the samples retain information about the hidden generator.  Unlike the standard QFT distribution, however, these peaks do not exhibit a simple number-theoretic pattern that can be directly converted into a continued-fraction estimate.  The irregular structure makes a compact analytical decoder difficult to construct.

An ideal generic decoder would be the maximum-likelihood estimator (MLE).  Let \(B_m=\{b_i\}_{i=1}^{m}\) denote \(m\) independent measurement outcomes, with \(b_i\in\{0,1\}^{n}\), and let \(\Pr(b\mid r)\) be the output probability of the HP circuit for a candidate period \(r\).  The MLE is
\begin{equation}
    \widehat r_{\mathrm{ML}}
    =
    \underset{r\in\mathcal R}{\operatorname{arg\,max}}\;
    \prod_{i=1}^{m}\Pr(b_i\mid r)
    =
    \underset{r\in\mathcal R}{\operatorname{arg\,max}}\;
    \sum_{i=1}^{m}\log \Pr(b_i\mid r),
    \label{eq:decoder_mle}
\end{equation}
where \(\mathcal R\) is the candidate-period set.  Under the usual regularity conditions for likelihood estimation, the Fisher information determines the local asymptotic precision of the estimator, so its inverse sets a local sample--precision scale~\cite{Paris2009,Kwon2026}. The samples we need for a target standard error \(\epsilon\) is
\begin{equation}
    m\sim \frac{1}{\fii(r,n)\,\epsilon^2}.
\end{equation}
For the discrete period problem, the global error exponent additionally depends on the separation between the full distributions \(\Pr(\cdot\mid r)\) for different \(r\); therefore the DFI alone should not be interpreted as a proof of an achievable decoding error exponent.  The empirical MLE failure rates shown in Fig.~\ref{fig:decoder_mle} decrease rapidly as the number of samples increases, and are approximately exponential over the displayed range.  This observation provides a useful numerical benchmark for the amount of information available in the HP samples, but it is not an asymptotic theorem for all periods or circuit families.

\begin{figure}[!h]
    \centering
    \includegraphics[width=0.9\linewidth]{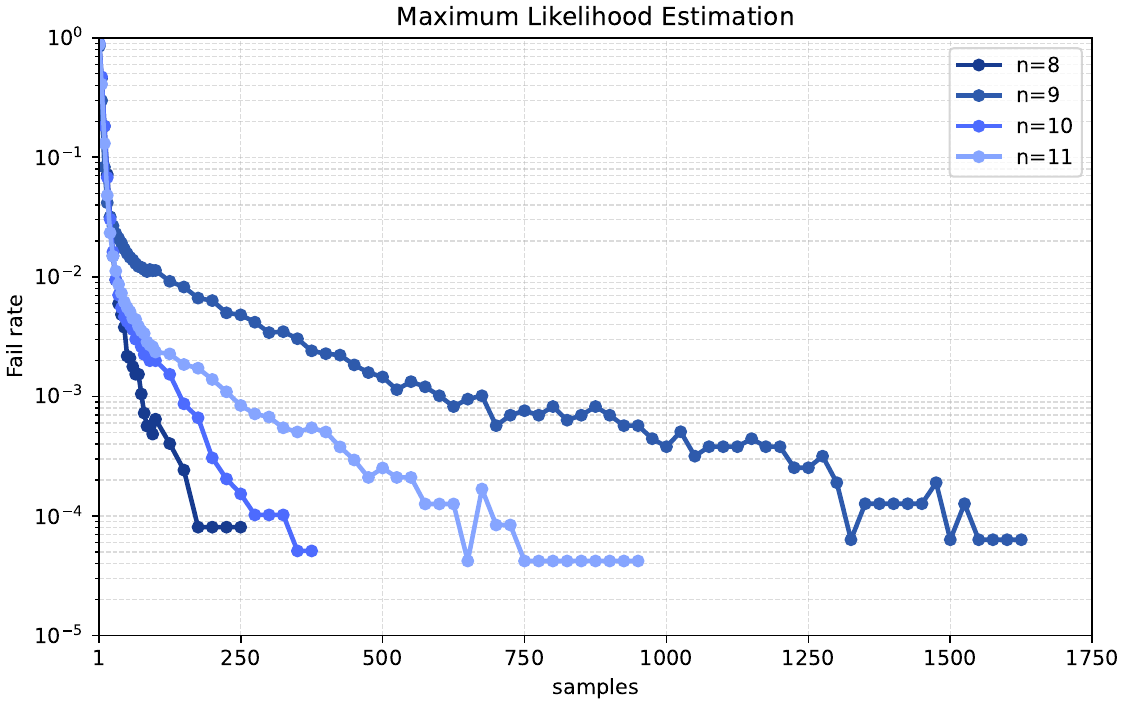}
    \caption{\textbf{Maximum-likelihood decoding benchmark.} Empirical failure
    rate of the MLE as a function of the number of measurement samples for
    \(n=8,9,10,\) and \(11\).  The rapid decrease in failure rate illustrates
    the improvement obtained from additional samples in the tested regime.}
    \label{fig:decoder_mle}
\end{figure}

The MLE is not the decoder used in the main numerical protocol.  Evaluating
Eq.~\eqref{eq:decoder_mle} requires computing \(\Pr(b_i\mid r)\) for every observed bit string and every candidate period.  A direct construction of the complete output distribution contains \(2^n\) probabilities and requires repeated circuit simulation, which is prohibitive in the regime of interest.  A specialized likelihood evaluation with lower cost is possible in principle, but its implementation and complexity are not established here.  We therefore replace the explicit likelihood calculation by a learned estimator that acts directly on the sampled bit strings.

The decoder used in the \(18\)-qubit experiments is a weighted, permutation-invariant Deep Sets model \cite{zaheer2017deep,edwards2016towards}.  We use this architecture because the measurement record is an unordered multiset: permuting the sampled bit strings must not change the decoded period.  It also avoids constructing an explicit frequency vector over the \(2^n\) possible bit strings or a separate logit for every possible period, either of which would make the decoder representation grow exponentially with the problem size.  Instead, it embeds only the observed strings and aggregates them into a fixed-dimensional representation.

Let \(d_e=\max(8,n)\) and \(d=16n\).  For bit position \(j\), the measured value \(b_{ij}\) and its position have independently learned embeddings \(E_{\mathrm{val}}(b_{ij}),E_{\mathrm{pos}}(j)\in\mathbb R^{d_e}\).  Their sum is concatenated over the \(n\) positions, layer-normalized, and projected from \(nd_e\) to \(d\) dimensions, followed by layer normalization and a GELU activation.  This representation is processed by three residual multilayer perceptrons (MLPs).  In the implementation, each MLP block applies layer normalization, a linear expansion from \(d\) to \(2d\), GELU, dropout, a linear projection back to \(d\), and a second dropout, then adds the result to its input through a residual connection:
\begin{equation}
    \mathcal B(h)
    =h+D\!\left[W_2D\!\left(\operatorname{GELU}
      \left(W_1\operatorname{LN}(h)+c_1\right)\right)+c_2\right],
    \qquad
    W_1:\mathbb R^d\to\mathbb R^{2d},\quad
    W_2:\mathbb R^{2d}\to\mathbb R^d,
    \label{eq:decoder_residual_block}
\end{equation}
where \(D\) denotes dropout.  The resulting shared map is denoted by
\(\phi:\{0,1\}^n\to\mathbb R^d\).

Repeated outcomes are stored once together with their counts.  If
\(\{\widetilde b_u\}_{u=1}^{U}\) are the distinct observed strings and \(c_u\)
are their multiplicities, the invariant representation is the count-weighted
mean
\begin{equation}
    z(B_m)=\sum_{u=1}^{U}w_u\phi(\widetilde b_u),
    \qquad
    w_u=\frac{c_u}{\sum_{v=1}^{U}c_v}.
    \label{eq:decoder_deep_sets}
\end{equation}
This is exactly the ordinary mean \(m^{-1}\sum_i\phi(b_i)\) without explicitly
repeating identical strings; zero-weight padding permits batching sets with
different numbers of distinct outcomes.  After another layer normalization,
two more residual blocks of the form in
Eq.~\eqref{eq:decoder_residual_block} act on \(z(B_m)\), producing the pooled
encoder representation \(h\).

During training, the target period is converted to its nibble sequence and the
network is optimized with teacher forcing.  If \(t_{1:T}(r)\) denotes the
sequence for period \(r\), the token loss is
\begin{equation}
    \mathcal L_{\mathrm{tok}}
    =-\frac{1}{T}\sum_{\ell=1}^{T}
      \log\operatorname{softmax}(q_\ell)_{t_\ell(r)},
    \label{eq:decoder_token_loss}
\end{equation}
with label smoothing in the numerical run.  At inference, beam search retains
at most four prefixes.  A prefix \(t_{1:\ell}\) is scored by the accumulated
conditional log-probability
\begin{equation}
    S(t_{1:\ell})
    =\sum_{j=1}^{\ell}
      \log\operatorname{softmax}(q_j)_{t_j},
    \label{eq:decoder_beam_score}
\end{equation}
and invalid prefixes outside the period range are removed before each beam pruning step.  The returned top-\(k\) list therefore has \(k\le4\), and its entries are decoded period integers, not indices of a 510-dimensional output layer.

\begin{figure}[!h]
    \centering
    \includegraphics[width=0.6\linewidth]{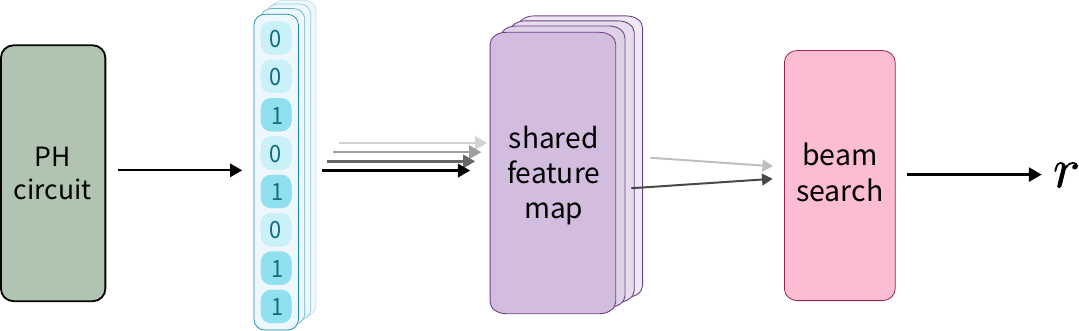}
    \caption{\textbf{Neural decoder architecture.} The HP circuit produces an
    unordered set of measurement bit strings.  A shared feature map processes
    every bit string, the resulting features are aggregated into a
    permutation-invariant representation, and beam search produces a ranked list of candidate periods.}
    \label{fig:decoder_architecture}
\end{figure}

The use of a neural decoder should be distinguished from an analytical decoding theorem.  The DFI characterizes the sensitivity of the HP output distribution to the hidden period and is therefore used here as an information-theoretic resolution scale~\cite{weimar2025fisher}.  Once the likelihood is replaced by a trained neural estimator, the DFI is not an asymptotically achieved bound for that estimator; the actual decoder performance must be established empirically through validation accuracy, top-\(k\) recovery, and the factorization tests.  In this sense, the neural network demonstrates that the information diagnosed by the DFI is practically accessible in the tested regime only.

\section{Numerical shift invariance for general periods}
\label{app:shift_invariant_zp}

The analytical result above establishes exact shift invariance for \(\mathbb{Z}_{2^n}\).  To test the finite arithmetic progressions associated with general periods \(r\), we numerically compare the standard \(2^n\)-point QFT with the fixed optimized \myqc-1 circuit.  We use \(n=18\) and periods \(r=1,\ldots,10\).  For each period, the shift is denoted by \(c\in\{0,\ldots,r-1\}\), and we use equal-cardinality truncated periodic input states
\begin{equation}
    \ket{\psi_{r,c}}
    =
    \frac{1}{\sqrt{R_r}}\sum_{q=0}^{R_r-1}\ket{c+qr},
    \qquad
    R_r=\left\lfloor\frac{2^n}{r}\right\rfloor .
    \label{eq:shift_invariance_periodic_state}
\end{equation}
For non-dyadic \(r\), these finite arithmetic progressions are not subgroups of \(\mathbb{Z}_{2^n}\), so this experiment probes the finite-period extension beyond the exact subgroup setting.

We quantify the shift dependence with the base-2 Jensen--Shannon divergence (JSD).  For a circuit \(U\), define \(\Pr_U(x\mid r,c)=|\langle x|U|\psi_{r,c}\rangle|^2\), where \(U\) is either the QFT or the fixed optimized \myqc-1 circuit.  We report \(\max_{0\le c<r}\operatorname{JSD}(\Pr_U(\cdot\mid r,0), \Pr_U(\cdot\mid r,c))\) for both circuits.  As a scale reference, the Uniform row reports the JSD between the unshifted \myqc-1 output distribution and the uniform distribution.  Identical distributions have JSD equal to zero, so smaller values indicate closer shift invariance.

\begin{table}[h]
    \color{revisioncolor}
    \centering
    \small
    \setlength{\tabcolsep}{4pt}
    \begin{tabular}{l *{10}{c}}
        \hline
        \(r\) & 1 & 2 & 3 & 4 & 5 & 6 & 7 & 8 & 9 & 10 \\
        \hline
        QFT (shift) & 0 & 0 & \(2.02\times10^{-6}\) & 0 & \(4.70\times10^{-6}\) & \(4.04\times10^{-6}\) & \(6.89\times10^{-6}\) & 0 & \(8.96\times10^{-6}\) & \(8.74\times10^{-6}\) \\
        \myqc-1 (shift) & 0 & 0 & 0.0041 & 0 & 0.0618 & 0.0040 & 0.0256 & 0 & 0.0324 & 0.0617 \\
        Uniform baseline & 1 & 0.999 & 0.781 & 0.999 & 0.726 & 0.769 & 0.742 & 0.997 & 0.724 & 0.715 \\
        \hline
    \end{tabular}
    \caption{\label{tab:shift_invariance_jsd}\textbf{Numerical shift-invariance check.}
    Base-2 JSD values at \(n=18\) for \(r=1,\ldots,10\).  The shift rows
    compare \(c=0\) with all \(0\le c<r\) and report the maximum.  The QFT
    entries are reported in scientific notation from the exact-support
    calculation.  The uniform baseline is the JSD between the unshifted
    \myqc-1 output and the uniform distribution.}
\end{table}

The QFT row is numerically small but nonzero, at approximately the \(10^{-6}\) scale in this finite-precision calculation.  Compared with the QFT, \myqc-1 is measurably farther from shift invariance for the non-dyadic periods, with a largest shift JSD of \(0.0618\), but these deviations remain small: the largest shift JSD is only \(8.63\%\) of the corresponding \myqc-1-to-uniform JSD.  Thus \myqc-1 is less shift-invariant than the QFT outside the exact \(\mathbb{Z}_{2^n}\) setting, while its absolute shift dependence remains modest in this numerical regime.

\section{Noise-robustness experiment}
\label{app:noise_robustness_details}

\begin{figure*}[h]
    \centering
    \includegraphics[width=0.78\textwidth]{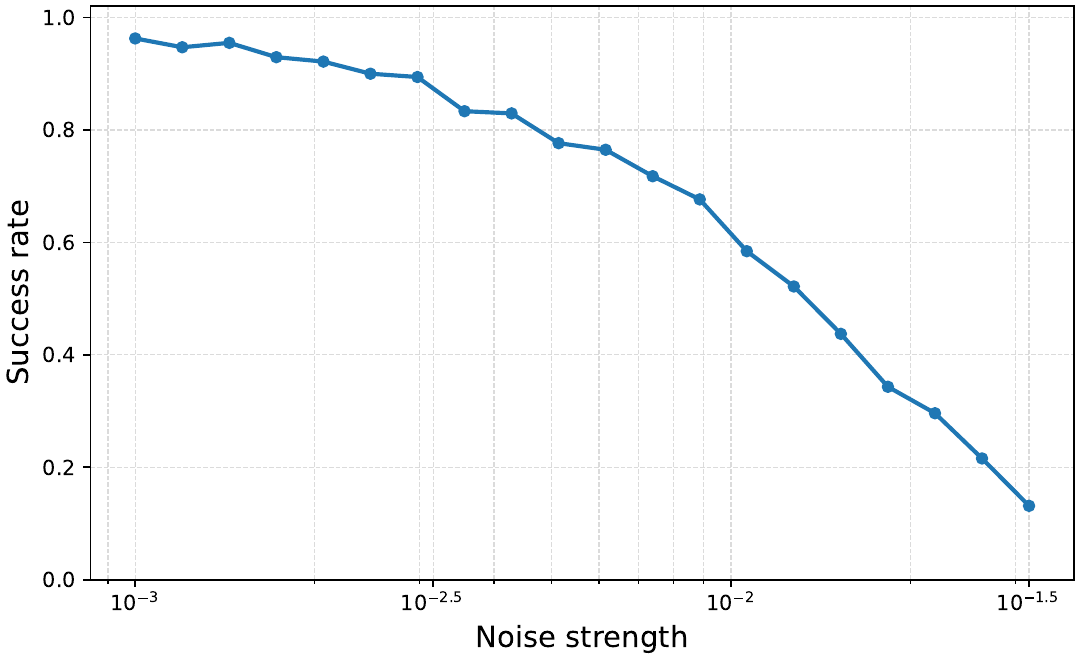}
    \caption{\textbf{Gate-level noise robustness.} Top-1 period-recovery accuracy of the fixed 18-qubit \myqc-1 neural decoder under stochastic Pauli errors inserted after every Hadamard and controlled-phase gate. The sweep evaluates \(510\) candidate periods, averages four noisy-circuit trajectories per period, and uses \(331776=1024n^2\) samples at \(n=18\).}
    \label{fig:gate_noise_accuracy}
\end{figure*}

We test the robustness of the trained \(18\)-qubit \myqc-1 neural-decoder
pipeline under gate-level stochastic Pauli noise. Throughout the noise
sweep, both the optimized circuit phases and the neural-network weights
are kept fixed, so that the results directly probe the robustness of the
trained pipeline without retraining or noise-aware calibration.

We model gate errors by independent local depolarizing channels applied
after each Hadamard and controlled-phase gate. Gate-local Pauli channels
provide a widely used effective description of noise in present-day
quantum processors. In particular, experiments on superconducting
devices have characterized gate noise using sparse local
Pauli--Lindblad models~\cite{vandenBerg2023}, and Pauli twirling or
randomized compiling can further tailor coherent errors into stochastic
Pauli channels~\cite{vandenBerg2023,Will2025}. In a recent Sycamore/Willow
experiment, for example, randomized-compiled CZ-gate errors were modeled
by two-qubit depolarizing channels with error rates obtained directly
from cross-entropy benchmarking~\cite{Will2025}.

Our choice of noise model is intentionally simpler than a device-specific experimental
noise model: a common noise strength \(\eta\) is used to provide a
controlled sweep rather than fitting separate error rates for individual
gates and qubits. More detailed hardware models additionally include
relaxation, readout errors, leakage, and correlated or crosstalk
errors~\cite{GoogleQAI2023,Bausch2024}. Indeed, comparisons with
superconducting-device data have shown that such effects can produce
correlations beyond those captured by a purely local depolarizing
model~\cite{GoogleQAI2023}. We therefore use the present model as a
standard gate-level robustness benchmark rather than as a complete
reconstruction of a particular hardware noise channel.

We evaluate the same \(510\) candidate periods \(r\in\{2,\ldots,511\}\) used in the period-recovery experiment. For each period and noise strength, the decoder receives \(1024n^2=331776\) measurement samples at \(n=18\). We consider \(20\) geometrically spaced noise strengths over
\begin{equation}
    \eta\in[10^{-3},10^{-1.5}],
\end{equation}
where the lower end probes the \(10^{-3}\)-level gate-error regime relevant to recent high-fidelity quantum processors, while the upper end extends deliberately into a stronger-noise regime to resolve the breakdown of period recovery~\cite{Evered2023,GoogleQAI2025}. The geometric spacing provides finer resolution in the low-noise regime. For each noise strength, we average over four noisy-circuit
realizations per period. 

The resulting top-1 period-recovery accuracy is shown in Fig.~\ref{fig:gate_noise_accuracy}. The accuracy decreases monotonically with increasing noise, from \(0.9059\) at \(\eta=10^{-3}\) to \(0.1412\) at \(\eta=10^{-1.5}\). It remains above \(0.84\) up to \(\eta=2.98\times10^{-3}\), and is \(0.7118\) at \(\eta=6.16\times10^{-3}\). These results show that the trained \myqc-1 decoder retains substantial period-recovery performance in the weak-noise regime, followed by a gradual degradation as local gate errors accumulate.

\end{document}